\documentclass{article}

\usepackage{tikz, pgf, pgfplots}
\usetikzlibrary
	{
		arrows, automata, chains, datavisualization.formats.functions, fit, positioning, shapes
	}

\usepackage{amsthm}
\usepackage[bb = boondox]{mathalfa}
\usepackage{dsfont}
\usepackage{mathtools}
\usepackage{multicol}
\usepackage{physics}

\usepackage[a4paper, left = 2.4cm, right = 2.4 cm, top = 2.5cm, bottom = 2.5 cm]{geometry}

\usepackage{float}
\usepackage{graphicx}
\usepackage{hyperref}
\hypersetup
{
	colorlinks, linkcolor = {WordRed!75!black},
	citecolor = {WordBlueDarker25}, urlcolor = {WordAquaDarker50}
}
\usepackage[boxed, linesnumbered, ruled, vlined]{algorithm2e}
\usepackage{spreadtab}
\usepackage{siunitx}					
\newtheorem{definition}{Definition}[section]
\newtheorem{proposition}{Proposition}[section]
\newtheorem{theorem}{Theorem}[section]
\newtheorem{corollary}{Corollary}[section]

\definecolor{WordBlueVeryLight}{RGB}{0, 176, 240}
\definecolor{MyBlue}{RGB}{0, 64, 128}
\definecolor{WordBlueDarker}{RGB}{31, 78, 121}
\definecolor{WordBlueDarker25}{RGB}{54, 96, 146}
\definecolor{WordBlueDarker50}{RGB}{36, 64, 98}
\definecolor{WordRed}{RGB}{255, 0, 102}
\definecolor{WordAquaLighter60}{RGB}{183, 222, 232}
\definecolor{WordIceBlue}{RGB}{223, 227, 229}
\definecolor{GreenTeal}{HTML}{008080}
\definecolor{MyDarkBlue}{RGB}{0, 51, 102}
\definecolor{MyLightRed}{RGB}{244, 213, 245}
\definecolor{WordVeryLightTeal}{RGB}{223, 236, 235}

\title{DKPRG or how to succeed in the Kolkata Paise Restaurant game via TSP}

\author{
	Kalliopi Kastampolidou, Christos Papalitsas and Theodore Andronikos \\
	Ionian University, \\
	Department of Informatics, \\
	7 Tsirigoti Square, Corfu, Greece\\
	Emails: \{c17kast, c14papa, andronikos\}@ionio.gr\\
}

\begin{document}

\maketitle

\begin{abstract}
	The Kolkata Paise Restaurant Problem is a challenging game, in which $n$ agents must decide where to have lunch during their lunch break. The game is very interesting because there are exactly $n$ restaurants and each restaurant can accommodate only one agent. If two or more agents happen to choose the same restaurant, only one gets served and the others have to return back to work hungry. In this paper we tackle this problem from an entirely new angle. We abolish certain implicit assumptions, which allows us to propose a novel strategy that results in greater utilization for the restaurants. We emphasize the spatially distributed nature of our approach, which, for the first time, perceives the locations of the restaurants as uniformly distributed in the entire city area. This critical change in perspective has profound ramifications in the topological layout of the restaurants, which now makes it completely realistic to assume that every agent has a second chance. Every agent now may visit, in case of failure, more than one restaurants, within the predefined time constraints. From the point of view of each agent, the situation now resembles more that of the iconic travelling salesman, who must compute an optimal route through $n$ cities. Following this shift in paradigm, we advocate the use of metaheuristics. This is because exact solutions of the TSP are prohibitively expensive, whereas metaheuristics produce near-optimal solutions in a short amount of time. Thus, via metaheuristics each agent can compute her own personalized solution, incorporating her preferences, and providing alternative destinations in case of successive failures. We analyze rigorously the resulting situation, proving probabilistic formulas that confirm the advantages of this policy and the increase in utilization. The detailed mathematical analysis of our scheme demonstrates that it can achieve utilization ranging from $0.85$ to $0.95$ from the first day, while rapidly attaining steady state utilization $1.0$. Finally, we note that the equations we derive generalize formulas that were previously presented in the literature, which can be seen as special cases of our results.

	\textbf{Keywords:}: Game theory, Kolkata Paise Restaurant Problem, TSP, metaheuristics, optimization, probabilistic analysis.
\end{abstract}

\section{Introduction} \label{sec:Introduction}

\subsection{The Kolkata Paise Restaurant Problem}

	The \emph{El Farol Bar} problem is a well-established problem in Game Theory. It was William Brian Arthur who introduced El Farol Bar problem in Inductive Reasoning and Bounded Rationality \cite{arthur1994inductive}. It cab be described as follows: $N$ people, the players, need to decide \emph{simultaneously} but \emph{independently} whether they will visit tonight a bar that offers live music. In order to have an enjoyable night the bar must not be too crowded. Each potential visitor does not know the number of attendances each night in advance, so the visitor must predict and decide whether she wants to go to the bar or stay home. Although the players decide using previous knowledge, their choice is not affected by previous visits and they cannot communicate with each other \cite{yeung2008minority}. In the El Farol Bar problem, the number of choices $n$ is equal to $2$, so the players have to choose between staying home or going out.

	The \emph{Kolkata Paise Restaurant Problem}, as well as the \emph{Minority Game}, are variants of the El Farol Bar problem. The Minority Game was first introduced in 1997 by Damien Challet and Yi-Cheng Zhang \cite{challet1997emergence}. They developed the mathematical formulation of the El Farol Bar which they named Minority Game. This game has an \emph{odd} number $N$ of agents and at each stage of the game they decide whether they will go to the bar or stay home. The minority wins and the majority loses. Agents have to decide whether they want to go to the bar or not, regardless of the predictions for the attendance size. The Minority Game is a binary symmetric version of the El Farol Bar problem, with the symmetry relying on the fact that the bar can contain half of the players.

	The Kolkata Paise Restaurant Problem (KPRP for short) is a repeated game that was named after the city Kolkata in India. In KPRP there are $n$ cheap restaurants (Paise Restaurants) and $N$ laborers who choose among these places for their quick lunch break. If the restaurant they go to is crowded, they have to return to work hungry, since they do not have time to visit another restaurant, or lack the resources needed to travel to another area. This generalization of the El Farol Bar is described as follows: each of a large number $N$ of laborers has to choose between a large number $n$ of restaurants, where usually $N = n$. In order for a player to win, that is to eat lunch, only one player should go to each restaurant. If more than one players attend the same restaurant at the same time, an agent is chosen randomly and only this agent is served. The player who gets to eat has a payoff equal to $1$, whereas all others who also chose this restaurant have a payoff equal to $0$. Each agent prefers to go to an unoccupied restaurant, than visit a restaurant where there are other agents as well. This realization in turn implies that the pure strategy Nash equilibria of the stage game are Pareto efficient. Consequently, there are exactly $n!$ pure strategy Nash equilibria for the stage game. This, combined with the rationality of the players, leads to the conclusion that it is possible to sustain a pure strategy Nash equilibrium of the stage game as a sub-game perfect equilibrium of the KPRP.

	In \cite{banerjee2013kolkata} each agent has a rational preference over the restaurants and, despite the fact that the first restaurant is the most preferred, all agents prefer to be served even at their least preferable restaurant than not to be served at all. The prices are considered to be identical and each restaurant is allowed to serve only one agent. If more than one laborers attend the same restaurant, one laborer is chosen randomly, while the others remain starved for that day. The Kolkata Paise Restaurant problem is symmetric, given the preferences of the agents over the set of restaurants. The game is non-trivial because there is a hierarchy among the restaurants, with the first being the most preferable. Another approach stipulates that if multiple agents choose the same restaurant they have to share the same meal and as a result, none of them is happy. The choice of each player is secret and they have to choose simultaneously. The players choose their strategy based on the payoffs. It is assumed that the restaurants charge their meals with the same price. There is even a version where some restaurants offer much tastier meals than others. This game is a repeated game with a period of one day, and the choices of each player are known to the other players at the end of the day. The agents have their personal strategy as to where they intend to have lunch. In order to attain the optimal solution, the agents have to communicate and coordinate their actions, something which is forbidden. As a result, some agents may end up hungry and, at the same time, some restaurants may waste their food.

	Some authors study the case where the number of restaurants $n$ is small and the agents take coordinated actions. Then, they analyze the game as a sub-game of KPRP and estimate the possibility to preserve the cyclically fair norm. As a result, punishment schemes need to be designed in this case. Every evening the agent makes up her mind based on her past experiences and the available information about each restaurant, which is supposed to be known to every agent. Each agent decides on her own, with no interaction with the other players. If more than one customers arrive at the same restaurant, an agent is randomly chosen to eat and the rest have to starve. There is a ranking system among the restaurants \emph{shared} by the customers. The $n!$ Pareto efficient states can be achieved when all customers get served. The probability of this event is very low, due to the absence of cooperation and disclosure among the agents.

\subsection{The Travelling Salesman Problem}

	In discrete optimization problems, the variables take discrete values and, usually, the objective is to find a graph or another similar visualization, 
	from an infinite or finite set \cite{bonnans2013perturbation}. 
	The Travelling Salesman Problem (TSP) is a famous optimization problem described as follows: a salesman has to visit all the nearby cities starting from a specific city to which the salesman must return \cite{Feillet:2005:TSP:1063813.1063851}. The only constraint is that the salesman must start and finish at the same specific city and visit each city only once. The visiting order is to be determined by the salesman each time the problem arises. The cities are connected through railway or roads and the cost of each travel is modeled by the difficulty in traversing the edges of the graph. The salesman has just one purpose and that is to visit all the cities with the minimum possible travel cost. In this problem, the optimum solution is the fastest, shortest and cheapest solution. TSP is easily expressed as a mathematical problem that typically assumes the form of a graph, where each of its nodes are the cities that the salesman has to visit. TSP was formulated during the 1800s by Sir William Rowan Hamilton and Thomas Kirkman and it was first studied by Karl Menger during the 1930s at Harvard and Vienna \cite{Feillet:2005:TSP:1063813.1063851}. The purpose of TSP is for the salesman to determine the route with the lowest possible cost. Some of the typical applications of TSP are network optimization and hardware identification problems. It has kept researchers busy for decades and many solutions have emerged. TSP is an NP-hard problem and the results of the practical, heuristic solutions are not always optimal, but approximate \cite{Feillet:2005:TSP:1063813.1063851}. The simplest ``naive'' solution to this problem is, of course, to try all possibilities and explore all paths, but the cost in time and complexity is so huge that is practically impossible. In order to overcome that, when solving a TSP the pragmatic focus is a near-optimal route, instead of always the best. For the graph depicted in Figure \ref{fig:A small scale instance of a TSP} the optimal tour is $1 \rightarrow 3 \rightarrow 4 \rightarrow 2 \rightarrow 1$ with cost $7 + 12 + 19 + 11 = 49$.

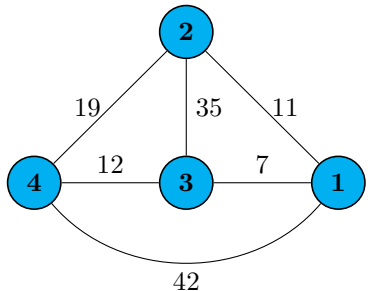
\begin{figure}[H] 
	\centering
	\begin{tikzpicture}
		[
		state/.style =
		{	
			circle,
			draw,
			fill = WordBlueVeryLight,
			minimum size = 7 mm,
			semithick
		},
		every loop/.style = {min distance = 12 mm},
		every edge/.style =										
		{	
			-,
			draw
		}
		]
		\node [state] (s1) at (2.0, 0.0) {\textbf{1}};
		\node [state] (s2) at (0.0, 2.0) {\textbf{2}};
		\node [state] (s3) at (0.0, 0.0) {\textbf{3}};
		\node [state] (s4) at (-2.0, 0.0) {\textbf{4}};
		\path [-]
		(s1) edge node [right] {$11$} (s2)
		(s1) edge node [above] {$7$} (s3)
		(s1) edge [bend left = 50] node [below] {$42$} (s4)
		(s2) edge node [right] {$35$} (s3)
		(s2) edge node [left] {$19$} (s4)
		(s3) edge node [above] {$12$} (s4);
	\end{tikzpicture}
	\caption{A small scale instance of a TSP.} \label{fig:A small scale instance of a TSP}
\end{figure}

\subsection{Related Work} \label{sec:Related}

\subsubsection{Games and the Kolkata Paise Restaurant Problem}

	The Kolkata Paise Restaurant Problem (KPRP) was initially introduced in an earlier form in 2007 \cite{chakrabarti2007kolkata}. Its current formulation appeared in 2009 in \cite{chakrabarti2009kolkata} and \cite{ghosh2010kolkata}. Subsequently, many creative ideas and different lines of thought have been published and even a quantum version of the game has arisen. In \cite{chakrabarti2009kolkata}, the importance of diversity is emphasized while herd behavior is penalized. Furthermore, the differences between the KPRP and the Minority Game are highlighted. One major difference is that in the KPRP the emphasis is placed on the \emph{simultaneous move many choice problem}, in contrast to the Minority Game, which studies a \emph{simultaneous move two choice problem}. Another important difference is the existence of a ranking system in the KPRP, but not in the Minority Game. Some of the strategies developed for the KPRP are discussed in \cite{Chakrabarti2017} which also discusses problems where these strategies can be successfully applied. Ghosh et al. in \cite{ghosh2010statistics} present a dictator's, or a social planner's as they call it, solution. In this solution the agents form a queue and the planner assigns each of them to a ranked restaurant depending on the queue of the first evening. The following evening the agents go to the next ranked restaurant and the last in the queue goes to the first ranked restaurant. This solution is called the \emph{fair social norm}. In real life, each agent decides in parallel or democratically every evening, so this solution may be considered somewhat unrealistic. However, the parallel decision or democratic decision strategy is not as efficient as the dictated one, with the last leading to one of the best solutions to this problem. Banerjee et al. in \cite{banerjee2013kolkata} offer a generalization of the problem in such a way that the cyclically fair norm is sustained. Each strategy is viewed as a sub-game of perfect equilibrium of the KPRP. In 2013, Ghosh et al. published an article about stochastic optimization strategies in the Minority Game and the KPRP \cite{ghosh2013kolkata}. There, they point out that a stochastic crowd avoiding strategy results in a efficient utilization in the KPRP. Reinforcement learning was first introduced in the KPRP by \cite{ghosh2017emergence}, together with six revision protocols aiming at efficient resource utilization. These protocols combine local information with reinforcement learning, Each revision protocol has two variants depending on whether or not customers who were once served by a restaurant remain loyal to that restaurant in all subsequent periods. Some of these protocols were experimentally tested and shown to improve the utilization rate. Another generalization was introduced by Yang et al. in \cite{yang2016mean} aiming at dynamic markets this time. They studied what happens when agents can either divert to another district or stay in the current one. Each agent may replace another agent with no prior knowledge of the game, following a Poisson distribution. Agarwal et al. in \cite{agarwal2016self} showed that the KPRP can be reduced to a Majority Game. In the latter, capacity is not restricted and agents aim at choosing with the herd. If more than one agents choose the same option, the utility decreases (see also \cite{milchtaich1996congestion} and \cite{Martin2019}). Abergel et al. in \cite{abergel2012econophysics} applied the KPRP in hospitals and beds. The local patients choose among the local hospitals those with the best ranking and compete with the other patients. If the patients are not treated in time it is a clear case of social waste of service for the rest of the hospitals. A brief presentation of the KPRP was given by Sharma et al. in \cite{sharma2017saga}, which included the origin and an overview of the game, strategies that may arise, several extensions and its applications in a variety of phenomena. The authors also presented an experimental analysis. Park et al. in \cite{park2017kolkata} introduced the KPRP in the Internet of Things (IoT) and IoT devices. They used a KPRP approach to develop a scheme for these devices, because it allowed them to model situations where multiple resources are shared among multiple users, each with individual preferences. In \cite{Sinha2020}, Sinha et al. propose a phase transition behavior, where if two or more agents visit the same restaurant, one is randomly picked to eat. The agents evolve their strategy based on the publicly available information about past choices in order for each of them to reach the best minority choice. In the same paper, they also develop two strategies for crowd-avoiding. 

	A significant trend, which has been quite evident in the last two decades, is to enhance classical games using unconventional means. The most prominent direction is to cast a classical game in a quantum setting. Since the pioneering works of Meyer \cite{Meyer1999} and Eisert et al. \cite{Eisert1999}, quantum versions for a plethora of well-known classical games have been studied in the literature. Starting from the most famous of all games, the Prisoners' Dilemma \cite{Eisert1999}, \cite{Li2014}, \cite{Deng2016}, \cite{Giannakis2019}, many researchers have sought to achieve better solutions by employing quantumness (see the recent \cite{Andronikos2018} and \cite{Giannakis2019} and references therein), or other tools, such as automata (\cite{Giannakis2015a}). It not surprising that unconventional approaches to classical games are undertaken because they promise clear advantages over the classical ones. Another line of research is to turn to biological systems for inspiration. The Prisoners' Dilemma features prominently in this setting also (see \cite{Kastampolidou2020} for a brief survey), but in reality most game situations can easily find analogues in biological and bio-inspired processes \cite{Theocharopoulou2019} and \cite{Kastampolidou2020a}. A quantum version of the KPRP was proposed in \cite{Sharif2012}, where the quantum Minority Game was expanded to a multiple choice version. The agents cannot communicate with each other and have to choose among $m$ choices, but an agent wins if she makes a unique choice. Higher payoffs than the classical version were observed due to shared entanglement and quantum operations. In Sharif's \cite{sharif2013introduction} review, quantum protocols for quantum games were introduced, including a protocol for a three-player quantum version of the KPRP. In \cite{ramzan2013three} the authors study the effect of quantum decoherence in a three-player quantum KPRP using tripartite entangled qutrit states. They observe that in the case of maximum decoherence the influence of the amplitude damping channel dominates over depolarizing and flipping channels. Furthermore, the Nash equilibrium of the problem does not change under decoherence.

\subsubsection{The Travelling Salesman Problem} \label{subsubsec:The TSP}

	The Travelling Salesman Problem is a well-known combinatorial optimization problem. In this problem a salesman must compute a route that begins from a particular node (the starting location), passes through all other nodes only once before returning to the starting location, and has the minimum cost. The first appearance of the term ``Travelling Salesman Problem'' probably occurred between 1931 and 1932. The core of the TSP problem, however, was first mentioned over a century before, in a 1832's German book \cite{voigt1981handlungsreisende}. The mathematical formulation was introduced by Hamilton and Kirkman \cite{voigt1981handlungsreisende} and is typically expressed as follows. A \emph{cycle} in a graph is a path that begins and ends at the same node and passes through all other nodes once. A \emph{Hamiltonian} cycle contains all the vertices of the graph. The Travelling Salesman Problem amounts to figuring the cheapest way to visit every city and return back. Research efforts on TSP and closely related problems include Ascheuer et al. \cite{Ascheuer2001} that addressed the asymmetric TSP-TW using more than three alternative integer programming formulations and more than ten neighborhood structures. Gutin and Punnen~\cite{gutin2006traveling} studied the effect of sorting-based initialization procedures. The authors claimed that understanding the algorithmic behavior is the best way to find solutions, since this would help in determining the best solution out of those available. Jones and Adamatzky~\cite{Jones2013} showed experimentally that using a sorting function within their algorithm was not functional and failed to return a feasible solution in some cases.

	The difficulty in tackling the TSP motivated researchers to explore other avenues. One such notable and particularly promising approach is based on metaheuristics. A metaheuristic is a high-level heuristic that is designed to recognize, build, or select a lower-level heuristic (such as a local search algorithm) that can provide a fairly good solution, particularly with missing or incomplete information or with limited computing capacity \cite{bianchi_survey_2009}. The term ``metaheuristics'' was coined by Glover. Metaheuristics can be used for a wide range of problems. Of course, it must be noted that metaheuristic procedures, in contrast to exact methods, do not guarantee a global optimal solution \cite{blum2003metaheuristics}. Papalitsas et al. \cite{Papalitsas2015} designed a metaheuristic based on VNS for the TSP with emphasis on Time Windows. Another quantum-inspired method, based on the original General Variable Neighborhood Search (GVNS), was proposed in order to solve the standard TSP \cite{Papalitsas2017}. This quantum-inspired procedure was also applied successfully to the solution of real-life problems that can be modeled as TSP instances \cite{Papalitsas2018}. A quantum-inspired procedure for solving the TSP with Time Windows was also presented in \cite{papalitsas2017b}. More recently, \cite{Papalitsas2019b} applied a quantum-inspired metaheuristic for tackling the practical problem of garbage collection with time windows that produced particularly promising experimental results, as further comparative analysis demonstrated in \cite{Papalitsas2019c}. A thorough statistical and computational analysis on asymmetric, symmetric, and national TSP benchmarks from the well known TSPLIB benchmark library, was conducted in \cite{Papalitsas2019a}. Very recently, Papalitsas et al. parameterized the TSPTW into the QUBO (Quadratic Unconstrained Binary Optimization) model \cite{Papalitsas2019}. The QUBO formulation enables TSPTW to run on a Quantum Annealer and is a critical step towards the ultimate goal of running the TSPTW with pure quantum optimization methods. Stochastic optimization can be implemented through several metaheuristic processes. The solution generated depends on the set of created random variables \cite{bianchi_survey_2009}. Metaheuristic processes may find successful solutions with less computational effort than accurate algorithms, iterative methods or basic heuristic procedures by looking for a wide variety of feasible solutions \cite{blum2003metaheuristics}. Hence, metaheuristic procedures are extremely useful and practical approaches for optimization because they can guarantee good solutions in a small amount of time. For example, a problem instance with thousands of nodes can be run for $30-40$ seconds and produce a solution with $3-5\%$ deviation from the optimal. This deviation depends on the implemented local search procedures inside the main part of the algorithm. An efficient design and choice of those improvement heuristics will define the deviation from the optimal solution. In view of the small amount of time they require and of the good quality of the solution they produce, we advocate their functional use in the Distributed Kolkata Paise Restaurant game.

\subsection{Contribution}

	Let us now briefly summarize the contributions of this paper.
	\begin{itemize}
		\item	We study the Kolkata Paise Restaurant Problem from an entirely new perspective. We identify and state explicitly certain implicit assumption that are inherent in the standard formulation of the game. We then take the unconventional step to abolish them entirely. This provides the opportunity for an entirely new setting and the adoption of a novel approach that leads to a new and more efficient strategy and, ultimately, to greater utilization for the restaurants.
		\item	For the first time, to the best of our knowledge, we focus on the spatial setting of the game and we propose a more realistic and plausible topological layout for the restaurants. We perceive the restaurants to be uniformly distributed in the entire city area. This, rather pragmatic and more probable in reality situation, has profound ramifications on the topological layout of the game: the restaurants now \emph{get closer} and, as their number $n$ increases, a standard assumption in the literature, the distances between nearby restaurants decrease. Due to the distribution of the restaurants, the resulting version of the game is aptly named the \emph{Distributed Kolkata Paise Restaurant Game}.
		\item	Thus, now it is realistic to assume that every agent has a second, a third, maybe even a fourth, chance. Every agent may visit, within the predefined time constraints, more than one restaurants. The agent is no longer a \emph{single destination and back} traveller. The agent now resembles the iconic \emph{travelling salesman}, who must pass through a network of cities, visiting every city once, coming back to the starting point, and all the time following the optimal route. This leads to the completely novel idea that each agent faces her own personalized TSP. We emphasize that the situation is specific for each agent, since the resulting network will vary. This is because each agent may have a different starting position and a different preference ranking of the restaurants. Of course, it is practically impossible to compute exact solutions for the TSP, as TSP is a famous NP-hard problem. However, this is a very small setback, as we may use metaheuristics. Metaheuristics can produce near-optimal solutions in a very short amount of time and this makes them indispensable tools of great practical value.
		\item	This entirely new setting is formalized and then rigorously analyzed via probabilistic tools. We derive general formulas that mathematically confirm the advantages of this policy and the increase in utilization. Detailed examples of typical instances of the game are given in a series of Tables and the derived equations are graphically depicted in order to demonstrate their qualitative and quantitative characteristics. Our scheme demonstrably achieves utilization ranging from $0.85$ and going to $0.95$ and even beyond from the first day. The steady state utilization, to which the game rapidly converges, is, as expected, $1.0$.
		\item	Finally, let us point out that the equations we derive generalize formulas that were previously presented in the literature, showing that the latter are actually special cases of our results.
	\end{itemize}

\subsection{Organization of the paper}

	The structure of this paper is as follows. In section \ref{sec:Introduction} we provide a  comprehensive description of the KPRP and the TSP. In subsection \ref{sec:Related} we mention some important works that deal with the KPRP and the TSP. The rigorous formulations of the KPRP and the TSP are presented in section \ref{sec:Background}. In section \ref{sec:DKPRG Formulation} we give a thorough explanation and presentation of the distributed version of the game, which we call Distributed Kolkata Paise Restaurant Game.
	We analyze mathematically the topological situation regarding the restaurants in section \ref{sec:Topological Considerations}, where the profound ramifications of the hypothesis that they follow the uniform probability distribution are developed. We formally prove the main results of the paper, which showcase the advantages of the distributed framework in a definitive manner in section \ref{sec:Mathematical Analysis} . Finally, in Section \ref{sec:Conclusion} we summarize our results and discuss future extensions of this work.

\section{Background} \label{sec:Background}

\subsection{Formulation of the standard Kolkata Paise Restaurant Problem} 

	In its most usual formulation, the Kolkata Paise Restaurant Problem is a repeated game with infinite rounds. There is a set of players, typically called \emph{agents} or \emph{customers}, that is denoted by $A = \{ a_1, \ldots, a_n \}$, a set of restaurants that is denoted by $R = \{ r_1, \ldots, r_n \}$, and a utility vector $u = ( u_1, \ldots, u_n ) \in \mathbb{R}^n$, which is associated with the restaurants and is common to every agent. On any given day, all agents decide to go to one of the $n$ restaurants for lunch. If it happens that just one agent arrives at a specific restaurant, then she will have lunch and she will be happy. If, however, two or more agents choose the same restaurant for lunch, then, it is generally assumed that just one of them will eat. The one to eat is chosen randomly. So, in such a case all but one will not be happy. Each agent has a utility and if they have lunch their utility is one, otherwise it is zero. In Chakrabarti et al. \cite{chakrabarti2009kolkata} the KPRP is modeled as a general one-shot restaurant game, where the set of agents is considered to be finite and the utilities are ranked as follows: $ 0 < u_n \leq \ldots \leq u_2 \leq u_1 $. The set of agents $A$ and the ranking of the utilities can be used to define the game. The latter can be represented as $G(u) = (A , S, \prod)$, where $A$ is the set of agents, $S$ is the set of strategies available to all agents, and $\prod = (\prod_1, \ldots, \prod_n )$ stands for the payoff vector. If the $i^{th}$ agent $a_i$ decides to go to the $j^{th}$ restaurant $r_j$, then the corresponding strategy is $s_i = j$. Every day each agent decides to which of the $n$ restaurants will go to eat. If $s_i = j$, this means that agent $a_i$ has decided to go to restaurant $r_j$. Given any strategy combination $s = (s_1, \ldots, s_n ) \in S^n$, the associated payoff vector is defined as $\prod ( s ) = ( \prod_1 ( s ), \ldots, \prod_n ( s ) )$, where the payoff $\prod_i ( s )$ of player $a_i$ is $\frac { u_{s_i} } { N_i (s) }$ and $N_i (s)$ is the total number of players that have made the same choice, i.e., restaurant $r_j$, as player $a_i$, including $a_i$. The strategy combination is in fact the restaurants the agents chose to eat to, and their payoff depends on their decision and the number of other agents that have made the same choice. In the literature, a game like KPRP, where there are potentially infinite rounds and in each round the same stage game is played, is referred to a \emph{supergame} \cite{mertens1989supergames}. A supergame is a situation where the same game is repeatedly played as a one-shot game and the agents count the payoff in the long run of the game. This makes the payoff function more complex due to the repetitions.

\subsection{Formulation of the TSP}

	The problem of finding the shortest Hamiltonian cycle is closely related to the TSP. The Hamiltonian graph problem, i.e., determining if a graph has a Hamiltonian cycle, is reducible to the traveling salesman problem. The trick is to assign zero length to the graph edges and, at the same time, create a new edge of length one for each missing edge. If the TSP solution for the resulting graph is zero, then there is a Hamiltonian cycle in the original graph; if the TSP solution is a positive number, then there is no Hamiltonian cycle in the original graph (see \cite{Lawler1985a}). In different fields, such as operational research and theoretical computer science, TSP, which is NP-hard, is of great significance. Usually TSP is represented by a graph. The fact that TSP is NP-hard implies that there is no known polynomial-time algorithm for finding an optimal solution regardless of the size of the problem instance \cite{Rego2011427}. There are two types of models for the TSP, \textit{symmetric} and \textit{asymmetric}. The former is represented by a complete undirected graph $G = (V, E)$ and the latter by a complete directed graph $G = (V, A)$. Assuming that $n$ denotes the number of cities (nodes), $V = \{ 1, 2, 3, \ldots, n \}$ is the set of vertices, $E = \{ (i, j) : i, j \in V, \text{ where } i < j \}$ is the set of edges, and $A = \{ (i, j) : i, j \in V, \text{ where } i \neq j \}$ is the set of arcs. A cost matrix $C = [ c_{i, j} ]$, which satisfies the triangle inequality $c_{i, j} \leq c_{i, k} + c_{k, j}$ for every $i, j, k$, is defined for each edge or arc. If $c_{i, j}$ is equal to $c_{j, i}$, the TSP is \emph{symmetric} (sTSP), otherwise it is called \emph{asymmetric} (aTSP). In particular, this is the case for problems where the vertices are points $P_i = (X_i, Y_i)$ of the Euclidean plane, and $c_{i, j} = \sqrt{ (X_j - X_i)^2 + (Y_j - Y_i)^2 }$ is the Euclidean distance. The triangle inequality holds if the quantity $c_{i, j}$ represents the length of the shortest path from $i$ to $j$ in the graph $G$ \cite{Johnson1989}. In the case of the symmetrical TSP, the number of all possible routes covering all cities and corresponding to all feasible solutions is given by $\frac {(n-1)!} {2}$ (recall that the number of cities is $n$). The cost of the route is the sum of the costs of the edges followed.

\section{Formulation of the DKPRG} \label{sec:DKPRG Formulation}

	The Kolkata Paise Restaurant problem (KPRP) is considered an extension of the minority game, as it involves multiple players ($n$) each having multiple choices ($N$). In its most general form it is possible that $n \neq N$. In this paper we follow the pretty much standard approach that the number of agents is equal to the number of restaurants, i.e., $n = N$. The novelty of our work lies on the fact that we advocate a spatially distributed and, in our view, more realistic version of the KPRP by taking into account the topology of the restaurants and by allowing the agents to begin their routes from different starting points. We call our version the \emph{Distributed Kolkata Paise Restaurant Game}, or DKPRG for from now on.

	In the original formulation of the KPRP one may readily point out the following important underlying assumptions.

	\begin{enumerate}
		\item[(\textbf{A1})]	All agents start from the same location.
		\item[(\textbf{A2})]	All restaurants are near enough to the point of origin of every customer, so that each customer can, in principle, go to any restaurant, eat there and return back to work in time, that is within the time window of the lunch break.
		\item[(\textbf{A3})]	Every restaurant is sufficiently far away from every other restaurant, so as to make prohibitive in terms of time constraints the possibility of any customer trying a second restaurant, in case her first choice proved fruitless.
	\end{enumerate}

	In the two dimensional setting of Kolkata, or, as a matter of fact, of any city, the above assumptions taken together imply something very close to the situation depicted in Figure \ref{fig:The topology of KPRP}. There, one can see that the agents are concentrated within a very narrow region, which can be viewed as the center of a conceptual ``circle.'' The restaurants are located on this ``circle'' and since \emph{no two of them are allowed to be close} they form something that resembles a ``regular polygon.''

\begin{figure}
	\centering
	\begin{tikzpicture} [scale = 1.25]
		\draw [line width = 1.5pt, MyBlue] (0, 0) circle [radius = 3cm];
		\node [rectangle, fill = WordRed] (Restaurants) at (0.0, 3.5){ \color{white} Restaurants };
		\node [rectangle, fill = WordRed] (r1) at (3.0, 0.0) { \color{white} $r_1$ };
		\node [rectangle, fill = WordRed] (r2) at ( { 3*cos(45) }, { 3*sin(45) } ) { \color{white} $r_2$ };
		\node [rectangle, fill = WordRed] (r3) at (0.0, 3.0){ \color{white} $r_3$ };
		\node [rectangle, fill = WordRed] () at ( { 3*cos(145) }, { 3*sin(145) } ) {};
		\node [rectangle, fill = WordRed] () at ( { 3*cos(180) }, { 3*sin(180) } ) {};
		\node [rectangle, fill = WordRed] () at ( { 3*cos(225) }, { 3*sin(225) } ) {};
		\node [rectangle, fill = WordRed] (rn-1) at (0.0, -3.0) { \color{white} $r_{n-1}$ };
		\node [rectangle, fill = WordRed] (rn) at ( { 3*cos(-45) }, { 3*sin(-45) } ) { \color{white} $r_n$ };
		\fill [WordAquaLighter60, line width = 1.5pt, rounded corners = 30pt](-1.5, -1.5) rectangle (1.5, 1.5);
		\node [rectangle, fill = WordBlueDarker50] (Agents) at (0.0, 0.0){ \color{white} Agents };
		\node [circle, fill = WordBlueDarker50] (a1) at ( { 1.3*cos(45) }, { 1.3*sin(45) } ) { \color{white}  ${\tiny a_1}$ };
		\node [circle, fill = WordBlueDarker50] () at ( { 1.3*cos(135) }, { 1.3*sin(135) } ) { };
		\node [circle, fill = WordBlueDarker50] () at ( { 1.3*cos(180) }, { 1.3*sin(180) } ) { };
		\node [circle, fill = WordBlueDarker50] () at ( { 1.3*cos(225) }, { 1.3*sin(225) } ) { };
		\node [circle, fill = WordBlueDarker50] (an) at ( { 1.3*cos(315) }, { 1.3*sin(315) } ) { \color{white}  ${\tiny a_n}$ };
	\end{tikzpicture}
	\caption{The assumed topology in the standard KPRP. The restaurants are located on a ``circle,'' forming a ``regular polygon.'' The agents are concentrated in a very narrow region around center of the ``circle.''} \label{fig:The topology of KPRP}
\end{figure}
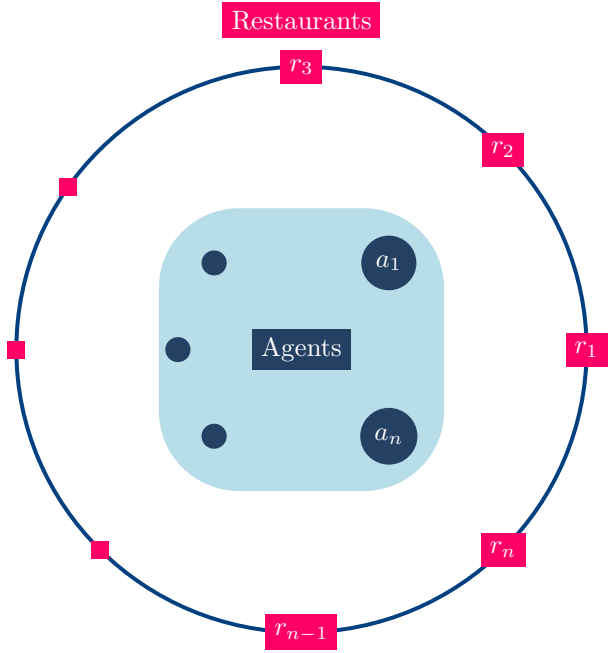

	This last remark is significant because it disallows a situation as the one shown in Figure \ref{fig:The forbidden topology of KPRP}. The spatial layout depicted in this Figure is strictly forbidden. The proximity of two, three or more restaurants would contradict the impossibility of a second chance. In the standard KPRP no agent is allowed a second chance. We write ``circle'' and ``regular polygon'' inside quotation marks because we are not obviously dealing with a perfect geometric circle or a perfect regular polygon, but two dimensional approximations resembling the aforementioned symmetric shapes. Clearly, this a very special topological layout, one that is highly unlikely to be observed in practice. There is no compelling reason for the restaurants to exhibit this regularity or the agents to be confined to approximately the same location. On the contrary, it would seem far more reasonable to assume that at least the restaurants and perhaps even the agents are uniformly distributed within a given area. Finally, the usual assumption that the preference ranking of the restaurants is common to all customers seems a bit too special and probably too restrictive.

\begin{figure}[H]
	\centering
	\begin{tikzpicture} [scale = 1.25]
		\draw [line width = 1.5pt, MyBlue] (0, 0) circle [radius = 3cm];
		\node [rectangle, fill = WordRed] (Restaurants) at (0.0, 3.5){ \color{white} Restaurants };
		\node [rectangle, fill = WordRed] (r1) at ( { 3*cos(0) }, { 3*sin(0) } ) { };
		\node [rectangle, fill = WordRed] (r1) at ( { 3*cos(35) }, { 3*sin(35) } ) { };
		\node [rectangle, fill = WordRed] (rj-1) at ( { 3*cos(70) }, { 3*sin(70) } ) { \color{white} $r_{j-1}$ };
		\node [rectangle, fill = WordRed] (rj) at ( { 3*cos(90) }, { 3*sin(90) } ) { \color{white} $r_j$ };
		\node [rectangle, fill = WordRed] (rj+1) at ( { 3*cos(110) }, { 3*sin(110) } ) { \color{white} $r_{j+1}$ };
		\node [rectangle, fill = WordRed] () at ( { 3*cos(145) }, { 3*sin(145) } ) {};
		\node [rectangle, fill = WordRed] () at ( { 3*cos(180) }, { 3*sin(180) } ) {};
		\node [rectangle, fill = WordRed] () at ( { 3*cos(225) }, { 3*sin(225) } ) {};
		\node [rectangle, fill = WordRed] (rn-1) at ( { 3*cos(260) }, { 3*sin(260) } ) { \color{white} $r_p$ };
		\node [rectangle, fill = WordRed] (rn-1) at ( { 3*cos(280) }, { 3*sin(280) } ) { \color{white} $r_q$ };
		\node [rectangle, fill = WordRed] (rn) at ( { 3*cos(-45) }, { 3*sin(-45) } ) { };
		\fill [WordAquaLighter60, line width = 1.5pt, rounded corners = 30pt](-1.5, -1.5) rectangle (1.5, 1.5);
		\node [rectangle, fill = WordBlueDarker50] (Agents) at (0.0, 0.0){ \color{white} Agents };
		\node [circle, fill = WordBlueDarker50] (a1) at ( { 1.3*cos(45) }, { 1.3*sin(45) } ) { \color{white}  ${\tiny a_1}$ };
		\node [circle, fill = WordBlueDarker50] () at ( { 1.3*cos(135) }, { 1.3*sin(135) } ) { };
		\node [circle, fill = WordBlueDarker50] () at ( { 1.3*cos(180) }, { 1.3*sin(180) } ) { };
		\node [circle, fill = WordBlueDarker50] () at ( { 1.3*cos(225) }, { 1.3*sin(225) } ) { };
		\node [circle, fill = WordBlueDarker50] (an) at ( { 1.3*cos(315) }, { 1.3*sin(315) } ) { \color{white}  ${\tiny a_n}$ };
	\end{tikzpicture}
	\caption{The spatial layout depicted in this Figure is strictly forbidden. The proximity of two, three or more restaurants would contradict the impossibility of a second chance. In the standard KPRP no agent is allowed a second chance.} \label{fig:The forbidden topology of KPRP}
\end{figure}
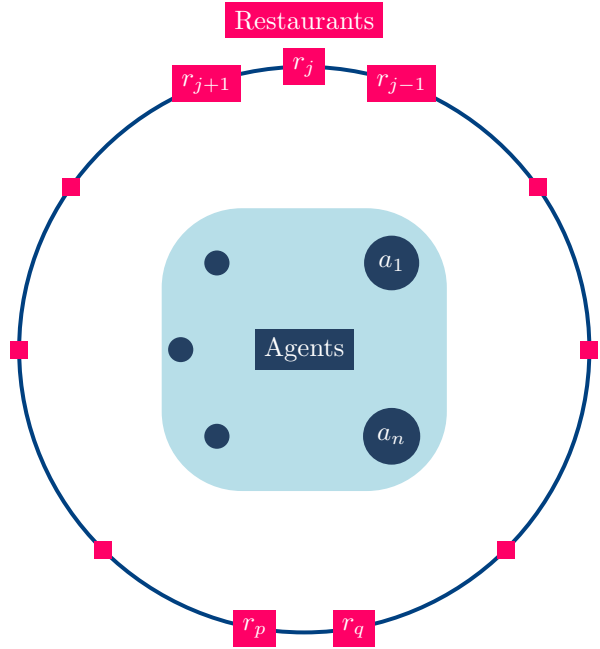

	With that motivation in mind, in this work we propose to abolish all these assumptions. The resulting game is spatially distributed in terms of restaurants and as such is called the Distributed Kolkata Paise Restaurant Game (DKPRG). In our setting, each customer may have her own staring point, which is, in general, different from the starting locations of the other customers. The staring locations can either be concentrated in a small region of the Kolkata city area, precisely like the standard KPRP, or they may be assumed to follow a random distribution. The fundamental difference with prior approaches is that now the restaurants are viewed as being \emph{uniformly distributed} over the city of Kolkata. This uniform randomness in the placement of the restaurants implies that there must be clusters of restaurants sufficiently near each other. This conclusion becomes inescapable, particularly in the case where the number of restaurants is large ($n \to \infty$). As will be shown in the following sections, the \emph{expected distance} between ``adjacent'' restaurants will be relatively short and will only decrease as the number $n$ of restaurants increases.

	The assumption of the random placement of restaurants leads to a \emph{personalized} situation for each individual agent: each agent is in effect faced with a \emph{personalized} Travelling Salesman Problem. To every agent corresponds an individual graph, which is assumed to be complete. The completeness assumption is not absolutely essential for the TSP, but, in any case, seems reasonable in the sense that one can go from any given restaurant to any other. This graph has $n + 1$ nodes, which are the locations of the $n$ restaurants plus the location of the starting point of the customer. The costs assigned to the edges of the graph are also \emph{personalized}; each agent combines an objective factor, the spatial distances between the restaurants, with a subjective factor, her personal preferences. Recall that in the DKPRG we forego the common preference restriction and we let every customer have a distinct preference, i.e., she may prefer a particular restaurant and dislike another. Let us clarify however, that getting served, even at the least preferable restaurant, is more desirable than not getting served at all! This in turn will lead to a possibly unique ordering of the restaurants from the most preferable to the least for \emph{each agent}. For instance, if two restaurants $r$ and $r'$ are equidistant from the staring point $s$ of a certain customer, something that is obviously an objective fact, but the customer in question has a clear preference for $r$ over $r'$, then she adjusts the costs $c_{s, r}$ and $c_{s, r'}$ corresponding to the edges $(s, r)$ and $(s, r')$, respectively, so that $c_{s, r} < c_{s, r'}$.

	Hence, every agent is faced with a distinctive network topology, which is the combined result of the inherent randomness of the spatial locations and the subjectiveness of her preferences. The topology of the restaurants has a further implication of the utmost importance: a customer whose first choice is a particular restaurant, will now have with very high probability the opportunity to visit a second, a third, or even a fourth restaurant in the same area if need be. For each agent the time cost is dominated by the time taken to visit the first restaurant; the trip to other nearby restaurants in the same region incurs a relatively negligible time cost due to their spatial proximity. The customer has a second (even a third) chance to be served within the time window of the lunch break. Thus, an efficient, if not optimal, method for every customer to make well-informed decisions regarding her first, second, third, etc. choice is to solve the Travelling Salesman Problem for her \emph{personalized} graph. Obviously, the TSP being an NP-hard problem, precludes the possibility of exact solutions. Nonetheless, near-optimal solutions of great practical value can easily be achieved in very short time by employing metaheuristics, as we have pointed out in subsection \ref{sec:Related}.

	From this perspective, we proceed now to propose an effective distributed strategy that, if adopted by every agent, will lead to an efficient global solution. All of them will use a common high-level strategy that is tailored and fine-tuned according to their individual preferences. To enhance clarity we explicitly state below the hypotheses and that define and characterize the DKPRG variant.

	\begin{enumerate}
		\item[(\textbf{H1})]	DKPRG is an infinitely repeated game.
		\item[(\textbf{H2})]	There are two main protagonists in the game. First, the $n$ \emph{agents} (also referred to as \emph{customers}) with different, in general, starting locations. The set of agents is denoted by $A = \{ a_1, \ldots, a_n \}$. Second, the $n$ \emph{restaurants} that are uniformly distributed within the same area. The set of restaurants is denoted by $R = \{ r_1, \ldots, r_n \}$. All agents know the locations of the restaurants, but each one of them need not know the starting locations of the other agents.
		\item[(\textbf{H3})]	To each agent $a \in A$ corresponds a distinct \emph{personal preference ordering} $P_a = (r_{j_1}, r_{j_2}, \ldots, r_{j_n})$ such that restaurant $r_{j_1}$ is her first preference, $r_{j_2}$ is her second preference, and so on, with $r_{j_n}$ being the least preferable restaurant for $a$.
		\item[(\textbf{H4})]	We adopt the standard convention that each restaurant can accommodate only one customer at a time. The immediate ramification of this convention is that if two or more customers arrive at a restaurant, only one can be served. The one to be served is chosen randomly.
		\item[(\textbf{H5})]	The aforementioned hypotheses immediately bring to the front the novelty and contribution of our approach. The positions of the restaurants relative to the starting point of each customer create for every customer a distinct topology, a distinct network of restaurants. Specifically, each agent $a \in A$ perceives a \emph{personalized graph} $G_a = (V_a, E_a)$. $G_a$ is a complete undirected graph having $n + 1$ vertices $v_0, v_1, \ldots, v_n$, where $v_0$ is the starting location of $a$ and $v_j$ is the location of restaurant $j, 1 \leq j \leq n$. For each pair of distinct vertices $u, v \in V_a$ there exists an undirected edge $(u, v) \in E_a$. The graph $G_a$ is complemented with the (symmetric) \emph{cost matrix} $C_a$, that assigns to each edge $(u, v)$ a cost $c_{u, v}$. We may surmise that the costs are computed by a function $f_a$ that incorporates geographical data, i.e., the distances between the restaurants, and the preference ordering $P_a$. The topological layout of the restaurants is an objective and global reality that is common to all customers and is undeniably crucial to a rational computation of the travel costs. On the other hand, it would be illogical if an agent did not take into account her preferences. The weight assigned to the spatial distances need not be equal to the weight assigned to the preferences. A conservative approach could assign a far greater weight to the distances compared to the preferences. A more idiosyncratic approach would deal with both on an equal footing by assigning equal weights to distances and preferences. It is plausible that for customer $a$ the personal preferences may play a more prominent role than for customer $a'$, in which case we may allow for the possibility that, in the process of computing the costs, each customer assigns completely different weights. In any event, we regard each cost matrix $C_a$ as distinct, which, along with the uniqueness of each $V_a$, explains why the resulting networks $G_a$ are all considered different, that is every customer is confronted with her own unique and \emph{personalized TSP}.
		\item[(\textbf{H6})]	Each customer $a \in A$ solves the corresponding TSP using an efficient metaheuristic that outputs a near-optimal solution. In that manner, $a$ computes a (near-optimal) \emph{tour} $T_a = ( l_0, l_1, \dots, l_n, l_{n+1} )$. The tour is represented by the ordered list $( l_0, l_1, \dots, l_n, l_{n+1} )$, where $l_0 = l_{n+1}$ is the starting point of $a$ and $l_k, 1 \leq k \leq n$, is the index of the restaurant in the $k^{th}$ position of the tour. Endowed with their individual route $T_a$, which is an integral part of their strategy, all customers follow a simple \emph{common} strategy. From their starting location $l_0$ they first travel to the restaurant $r_{l_1}$. Once there, those that get served conclude their route successfully. Those that do not get served, proceed to the restaurant $r_{l_2}$. If their attempt at getting lunch also fails at $r_{l_2}$, then they proceed to $r_{l_3}$, and so on. Obviously, the time constraints, that is the fact that the agent must have returned to her staring point by the time the lunch break is over, means that the agent will not have the opportunity to exhaust the entire tour. The customer must interrupt the tour at some point in order to return. This may happen after travelling unsuccessfully to two, three, or more restaurants, depending on the topology of the network. 
		\item[(\textbf{H7})]	\textbf{The Revision Strategy}. We adopt the standard assumption that the agents operate independently and no communication takes place between any two of them. Therefore, each customer is completely unaware of the routes of the other customers. They revise their strategy every evening taking into account only what happened during the present day. This means that they decide using only information from the last day and no prior information or history need to be kept. We assume that all agents follow the same policy. If they got served at a specific restaurant this day, then tomorrow they go straight to the \emph{same} restaurant. This applies even if this restaurant is not in the first place of their tour. For example, those agents that failed to get lunch at their first choice, but managed to do so at their second, or third choice, tomorrow go straight the restaurant that served them despite the fact that this particular restaurant is not their most preferable. Those that failed to get lunch, only know \emph{which restaurants were left vacant}, i.e., not visited by any agent today. Further or more elaborate information, such as the choices of other players or if they got served and at which restaurant, seems unnecessary. The unserved agents construct and solve their new personalized TSP, this time using as vertices only the \emph{vacant restaurants} (plus of course their starting location).
	\end{enumerate}

	Having explained the details of the DKPRG, we shall proceed to analyze the mathematical characteristics and evaluate the resulting utilization of this policy in the following sections.

\section{Topological considerations} \label{sec:Topological Considerations}

	We begin this section by fixing the notation and giving some definitions to clarify the most important concepts of our exposition.

	\begin{definition} \
		\begin{itemize}
			\item   The one-shot DKPRG takes place every day. We use the parameter $t = 1, 2, \ldots,$ to designate the day under consideration.
			\item   To each agent $a \in A$ corresponds the \emph{personalized network} $G_a = (V_a, E_a)$ together with the \emph{personalized cost matrix} $C_a$, which are constructed in the way we outlined in the previous section. Agent $a$ follows the tour $T_a = ( l_0, l_1, \dots, l_n, l_{n+1} )$, which is the solution to her \emph{personalized TSP}. As we have emphasized, by using metaheuristics it is possible to obtain near-optimal solutions in a very short amount of time. 
			\item   The quality and efficiency of the strategy is measured by the \emph{utilization} ratio $f$. This is of course the fraction of agents being served in a day, or, equivalently, the fraction of restaurants serving customers in a day. The equivalence is obvious because there are $n$ customers and $n$ restaurants.
		\end{itemize}
	\end{definition}

	In section \ref{sec:Mathematical Analysis} we shall revisit the concept of utilization and we shall be more precise by asserting the expected utilization per day as a function of the game parameters.

	An agent $a$ who has opted to follow tour $T_{a} = ( l_0, l_1, \dots, l_n, l_{n+1} )$ will initially try to get lunch at restaurant $r_{l_1}$. If she succeeds, she will eat and then return to her starting point. If she fails, she will visit the next restaurant in the tour, i.e., $r_{l_2}$. If she gets lunch there, she will subsequently go back to work. This process will go on until either she gets served or runs out of time, in which case she must interrupt the tour and return to work. If the time constraints allow her to pass through the \emph{first} $m$ restaurants in her tour, in the worst-case scenario of $m - 1$ consecutive failures, then we say that $T_{a}$ is an $m$-\emph{stop} tour. To facilitate our mathematical analysis we take for granted that \emph{all} customers follow $m$-stop tours. We have already explained why, in our view, $m$ must be $\geq 2$. The case where $m = 1$ reduces to the standard treatment of the KPRP, which has already been analyzed extensively in the literature. In the rest of this work we study the case where $m \geq 2$. All these considerations motivate the next definition.

	\begin{definition} \
		\begin{itemize}
			\item   The tour $T_{a} = ( l_0, l_1, \dots, l_n, l_{n+1} )$ associated with agent $a$ is an $m$-\emph{stop} tour, $m \geq 2$, if, in the worst case, agent $a$ can visit restaurants $l_1, l_2, \ldots, l_m$ in this order without violating her time constraints. In such a tour, $l_1$ is the \emph{first} stop, $l_2$ is the \emph{second} stop, and so on, with $l_m$ being the final $m^{th}$ stop. 
			\item   If $\forall a \in A$, $T_{a}$ is an $m$-stop tour, then the resulting game is the $m$-stop DKPRG.
		\end{itemize}
	\end{definition}

	Let us now explore the spatial ramifications of our assumption that the restaurants are uniformly distributed within the overall city area. We now give the formal definition of uniform distribution.

	\begin{definition} \label{def:Uniform Distribution}
		Given a region $B$ on the plane, a random variable $L$ has \emph{uniform distribution} on $B$, if given any subregion\footnote{If one wants to be overly technical, one should assume that both $B$ and $C$ are measurable sets. In the current setting, we believe that it is unnecessary to go to such a technical depth.} $C$ the following holds:

		\begin{align} \label{eq:Uniform Distribution Definition}
			P(L \in C) = \frac{ \text{area} ( C ) } { \text{area} ( B ) }, \quad C \subset B \ .
		\end{align}

		We assume of course that $L$ takes values in $B$.
	\end{definition}

	The above definition is adapted from \cite{Pitman1993}. For a more general and sophisticated definition in terms of measures we refer the interested reader to \cite{Klenke2014}.

	\begin{proposition} \label{thr: Expected Number of Restaurants}
		Assuming that the $n$ restaurants are uniformly distributed on the whole city area, then if the city area is partitioned into $n$ regions of equal area, the \emph{expected number} $\overline{N_p}$ of restaurants in each region is exactly $1$.

		\begin{align} \label{eq:Expected Value Np}
			\overline{N_p} = 1 \ , \ 1 \leq p \leq n \ . 
		\end{align}

	\end{proposition}

	\begin{proof}
		Let $B$ stand for the whole city area and let $B_1, \ldots, B_n$ be the $n$ regions. The hypotheses assert that:

		\begin{enumerate}
			\item	$B_1 \cup \ldots \cup B_n = B$,
			\item	$B_p \cap B_q = \emptyset$, if $1 \leq p \neq q \leq n$, and
			\item	$\text{area} ( B_1 ) = \text{area} ( B_2 ) = \ldots = \text{area} ( B_n ) = \frac {\text{area} ( B ) } { n }$.
		\end{enumerate}

		Invoking the fact that the $n$ restaurants are uniformly distributed on the whole city, we deduce from (\ref{eq:Uniform Distribution Definition}) that for every restaurant $r_j, 1 \leq j \leq n$, and for every region $B_p, 1 \leq p \leq n$,

		\begin{align} \label{eq:Probability rj In Bp}
			P(r_j \in B_p) = \frac{ \text{area} ( B_p ) } { \text{area} ( B ) } = \frac { 1 } { n }
			\ , \quad 1 \leq j, p \leq n \ .
			\tag{ \ref{thr: Expected Number of Restaurants}.i }
		\end{align}

		We may now define the following collection of auxiliary random variables $N_{pj}$, where $1 \leq p, j \leq n$.

		\begin{align} \label{eq:Random Variables Npj}
			N_{p j} =
			\left\{
			\begin{matrix*}[l]
				1 & \text{if restaurant } r_j \text{ is located in region } B_p \\
				0 & \text{otherwise}
			\end{matrix*}
			\right.
			\ .
			\tag{ \ref{thr: Expected Number of Restaurants}.ii }
		\end{align}

		By combining the result from (\ref{eq:Probability rj In Bp}) with definition (\ref{eq:Random Variables Npj}), we may conclude that

		\begin{align} \label{eq:Probability of Random Variables Npj}
			N_{p j} =
			\left\{
			\begin{matrix*}[l]
				1 & \text{with probability } \frac{ 1 } { n } \\
				0 & \text{with probability } \frac{ n - 1 } { n }
			\end{matrix*} 
			\right.
			\ , \quad 1 \leq p, j \leq n \ .
			\tag{ \ref{thr: Expected Number of Restaurants}.iii }
		\end{align}

		Then, the random variable

		\begin{align} \label{eq:Random Variables Np}
			N_p = \sum_{ j = 1 }^{ n } N_{p j} \quad ( 1 \leq p \leq n ) \tag{ \ref{thr: Expected Number of Restaurants}.iv }
		\end{align}

		gives the number of restaurants in region $B_p$, $1 \leq p \leq n$. We are not interested in the actual value of the random variable $N_p$ per se, but in its expected value $E \left[ N_p \right]$. The latter can be easily computed if we use the above results and the linearity of the expected value operator.

		\begin{align} \label{eq:Expected Value of Random Variables Np}
			\overline{N_p} = E \left[ N_{ p } \right]
			\overset{ (\ref{eq:Random Variables Np}) } { = }
			E \left[ \sum_{ j = 1 }^{ n } N_{p j} \right] =
			\sum_{ j = 1 }^{ n } E \left[ N_{p j} \right]
			\overset{ (\ref{eq:Probability of Random Variables Npj}) } { = }
			\sum_{ j = 1 }^{ n } \left( 1 \cdot \frac{ 1 }{ n } \right) =
			1  \tag{ \ref{thr: Expected Number of Restaurants}.v }
		\end{align}
	
		This establishes that the expected number of restaurants in each region is precisely $1$ and proves formula (\ref{eq:Expected Value Np}).
	\end{proof}

	Partitioning a city area into $n$ disjoint regions of equal area might not be an easy task. The point is that for large values of $n$, as is the standard assumption in the literature, it is certainly doable. We stress the fact the shape of the regions need not be the same. Indeed, the validity of Proposition \ref{thr: Expected Number of Restaurants} holds irrespective of whether the regions have the same shape or any particular shape for that matter.

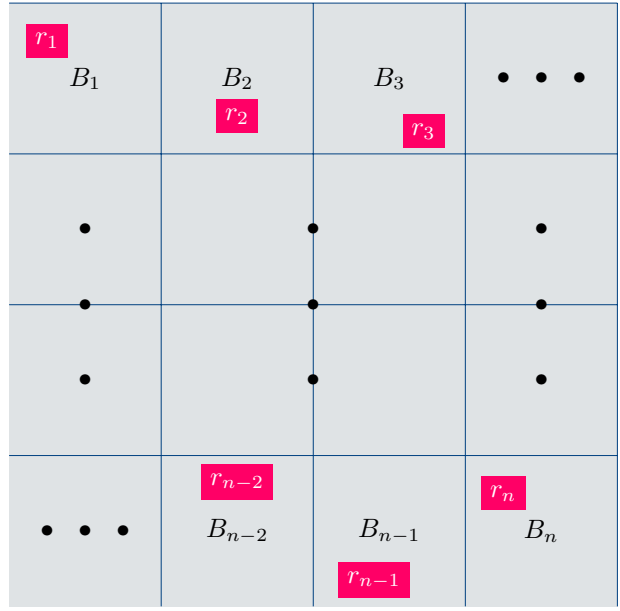
\begin{figure}
	\centering
	\begin{tikzpicture} [scale = 1.0]
		\fill [WordIceBlue] (-4.0, -4.0) rectangle (4.0, 4.0);
		\draw [step = 2.0 cm, MyBlue, line width = 0.25pt] (-4.0, -4.0) grid (4.0, 4.0);
		\node [rectangle, minimum size = 2.0 cm] () at (-3.0, 3.0) { $B_1$ };
		\node [rectangle, minimum size = 2.0 cm] () at (-1.0, 3.0) { $B_2$ };
		\node [rectangle, minimum size = 2.0 cm] () at (1.0, 3.0) { $B_3$ };
		\node [rectangle, minimum size = 2.0 cm] () at (-1.0, -3.0) { $B_{n-2}$ };
		\node [rectangle, minimum size = 2.0 cm] () at (1.0, -3.0) { $B_{n-1}$ };
		\node [rectangle, minimum size = 2.0 cm] () at (3.0, -3.0) { $B_n$ };
		\node [rectangle] at (2.5, 3.0){ $\bullet$ };
		\node [rectangle] at (3.0, 3.0){ $\bullet$ };
		\node [rectangle] at (3.5, 3.0){ $\bullet$ };
		\node [rectangle] at (-3.0, 1.0){ $\bullet$ };
		\node [rectangle] at (-3.0, 0.0){ $\bullet$ };
		\node [rectangle] at (-3.0, -1.0){ $\bullet$ };
		\node [rectangle] at (0.0, 1.0){ $\bullet$ };
		\node [rectangle] at (0.0, 0.0){ $\bullet$ };
		\node [rectangle] at (0.0, -1.0){ $\bullet$ };
		\node [rectangle] at (3.0, 1.0){ $\bullet$ };
		\node [rectangle] at (3.0, 0.0){ $\bullet$ };
		\node [rectangle] at (3.0, -1.0){ $\bullet$ };
		\node [rectangle] at (-3.5, -3.0){ $\bullet$ };
		\node [rectangle] at (-3.0, -3.0){ $\bullet$ };
		\node [rectangle] at (-2.5, -3.0){ $\bullet$ };
		\node [rectangle, fill = WordRed] (Restaurants) at (0.0, 4.5){ \color{white} Uniformly Distributed Restaurants };
		\node [rectangle, fill = WordRed] (r1) at (-3.5, 3.5) { \color{white} $r_1$ };
		\node [rectangle, fill = WordRed] (r2) at (-1.0, 2.5) { \color{white} $r_2$ };
		\node [rectangle, fill = WordRed] (r3) at (1.45, 2.3) { \color{white} $r_3$ };
		\node [rectangle, fill = WordRed] (rn-1) at (-1.0, -2.35) { \color{white} $r_{n-2}$ };
		\node [rectangle, fill = WordRed] (rn-1) at (0.8, -3.65) { \color{white} $r_{n-1}$ };
		\node [rectangle, fill = WordRed] (rn) at (2.5, -2.5) { \color{white} $r_n$ };
	\end{tikzpicture}
	\caption{Kolkata can be conceptually partitioned into $n$ regions $B_1, \ldots, B_n$ of \emph{equal} area. If the $n$ restaurants are uniformly distributed in the overall Kolkata area, then the \emph{expected number} of restaurants in each region $B_j, 1 \leq j \leq n$, is $1$.} \label{fig:Uniformly Distributed Restaurants}
\end{figure}

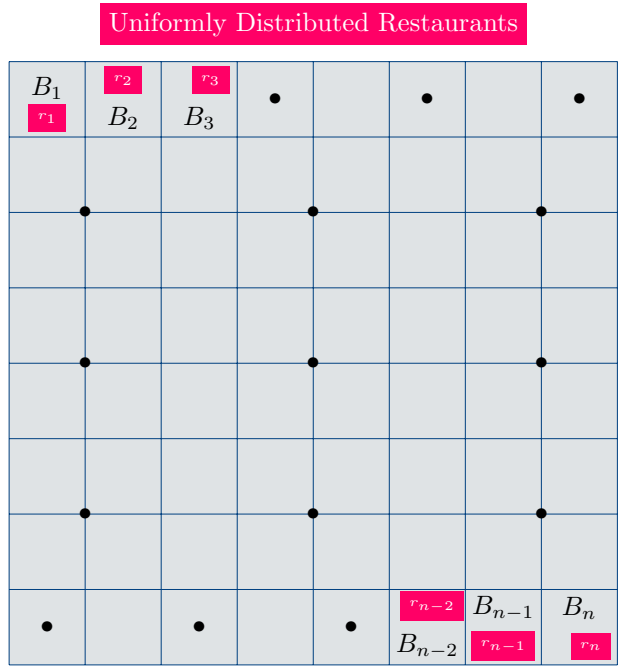
\begin{figure}
	\centering
	\begin{tikzpicture} [scale = 1.0]
		\fill [WordIceBlue] (-4.0, -4.0) rectangle (4.0, 4.0);
		\draw [step = 1.0 cm, MyBlue, line width = 0.25pt] (-4.0, -4.0) grid (4.0, 4.0);
		\node [rectangle, minimum size = 2.0 cm] () at (-3.5, 3.65) { $B_1$ };
		\node [rectangle, minimum size = 2.0 cm] () at (-2.5, 3.25) { $B_2$ };
		\node [rectangle, minimum size = 2.0 cm] () at (-1.5, 3.25) { $B_3$ };
		\node [rectangle, minimum size = 2.0 cm] () at (1.5, -3.75) { $B_{n-2}$ };
		\node [rectangle, minimum size = 2.0 cm] () at (2.5, -3.25) { $B_{n-1}$ };
		\node [rectangle, minimum size = 2.0 cm] () at (3.5, -3.25) { $B_n$ };
		\node [rectangle] at (-0.5, 3.5){ $\bullet$ };
		\node [rectangle] at (1.5, 3.5){ $\bullet$ };
		\node [rectangle] at (3.5, 3.5){ $\bullet$ };
		\node [rectangle] at (-3.0, 2.0){ $\bullet$ };
		\node [rectangle] at (-3.0, 0.0){ $\bullet$ };
		\node [rectangle] at (-3.0, -2.0){ $\bullet$ };
		\node [rectangle] at (0.0, 2.0){ $\bullet$ };
		\node [rectangle] at (0.0, 0.0){ $\bullet$ };
		\node [rectangle] at (0.0, -2.0){ $\bullet$ };
		\node [rectangle] at (3.0, 2.0){ $\bullet$ };
		\node [rectangle] at (3.0, 0.0){ $\bullet$ };
		\node [rectangle] at (3.0, -2.0){ $\bullet$ };
		\node [rectangle] at (-3.5, -3.5){ $\bullet$ };
		\node [rectangle] at (-1.5, -3.5){ $\bullet$ };
		\node [rectangle] at (0.5, -3.5){ $\bullet$ };
		\node [rectangle, fill = WordRed] (Restaurants) at (0.0, 4.5){ \color{white} Uniformly Distributed Restaurants };
		\node [rectangle, fill = WordRed] (r1) at (-3.5, 3.25) { \color{white} \tiny $r_1$ };
		\node [rectangle, fill = WordRed] (r2) at (-2.5, 3.75) { \color{white} \tiny $r_2$ };
		\node [rectangle, fill = WordRed] (r3) at (-1.35, 3.75) { \color{white} \tiny $r_3$ };
		\node [rectangle, fill = WordRed] (rn-1) at (1.56, -3.22) { \color{white} \tiny $r_{n-2}$ };
		\node [rectangle, fill = WordRed] (rn-1) at (2.5, -3.75) { \color{white} \tiny $r_{n-1}$ };
		\node [rectangle, fill = WordRed] (rn) at (3.65, -3.75) { \color{white} \tiny $r_n$ };
	\end{tikzpicture}
	\caption{As $n \to \infty$, the expected number of restaurants in each region remains $1$, but the expected distance between restaurants in neighbouring regions decreases.} \label{fig:Uniformly Distributed Restaurants for Large n}
\end{figure}

	This topological layout of the restaurants is shown in Figures \ref{fig:Uniformly Distributed Restaurants} and \ref{fig:Uniformly Distributed Restaurants for Large n}. In these Figures, the regions are drawn are squares, but this is just for convenience and to facilitate their graphic depiction. As we have explained, the regions are not required to have the same shape and nor does their shape need to resemble a regular two dimensional figure. For very large values of $n$, partitioning a city into very small identical squares is a good approximation, as we know from the field of image representation.

	It is useful to contrast the two Figures. The latter depicts the situation where the number of restaurants is much larger compared to the number of restaurants in the former Figure. This demonstrates clearly what happens when $n$ increases significantly, i.e., when $n \to \infty$. Irrespective of the size of magnitude of $n$, the expected number of restaurants in each of the $n$ regions (recall that they are pairwise disjoint and of equal area) remains $1$. What does change however is the \emph{area} of each region, which \emph{decreases} with $n$ and, as a consequence, the expected \emph{distance} between restaurants located in adjacent regions.

\begin{figure}
	\centering
	\begin{tikzpicture} [scale = 1.0]
		\fill [WordIceBlue] (-4.0, -4.0) rectangle (4.0, 4.0);
		\draw [step = 2.0 cm, MyBlue, line width = 0.1pt] (-4.0, -4.0) grid (4.0, 4.0);
		\node [rectangle, minimum size = 2.0 cm] () at (-2.5, 2.5) { \footnotesize $B_1$ };
		\node [rectangle, minimum size = 2.0 cm] () at (-1.5, 2.85) { \footnotesize $B_2$ };
		\node [rectangle, minimum size = 2.0 cm] () at (1.0, 3.5) { \footnotesize $B_3$ };
		\node [rectangle, minimum size = 2.0 cm] () at (-0.45, -2.9) { \footnotesize $B_{n \! - \! 2}$ };
		\node [rectangle, minimum size = 2.0 cm] () at (1.4, -3.07) { \footnotesize $B_{n \! - \! 1}$ };
		\node [rectangle, minimum size = 2.0 cm] () at (3.0, -2.5) { \footnotesize $B_n$ };
		\node [rectangle] at (2.5, 3.0){ \Large $\cdot$ };
		\node [rectangle] at (3.0, 3.0){ \Large $\cdot$ };
		\node [rectangle] at (3.5, 3.0){ \Large $\cdot$ };
		\node [rectangle] at (-3.0, 1.0){ \Large $\cdot$ };
		\node [rectangle] at (-3.0, 0.0){ \Large $\cdot$ };
		\node [rectangle] at (-3.0, -1.0){ \Large $\cdot$ };
		\node [rectangle] at (0.0, 1.0){ \Large $\cdot$ };
		\node [rectangle] at (0.0, 0.0){ \Large $\cdot$ };
		\node [rectangle] at (0.0, -1.0){ \Large $\cdot$ };
		\node [rectangle] at (3.0, 1.0){ \Large $\cdot$ };
		\node [rectangle] at (3.0, 0.0){ \Large $\cdot$ };
		\node [rectangle] at (3.0, -1.0){ \Large $\cdot$ };
		\node [rectangle] at (-3.5, -3.0){ \Large $\cdot$ };
		\node [rectangle] at (-3.0, -3.0){ \Large $\cdot$ };
		\node [rectangle] at (-2.5, -3.0){ \Large $\cdot$ };
		\node [rectangle, fill = WordRed] (r1) at (-3.65, 3.75) { \color{white} \tiny $r_1$ };
		\node [rectangle, fill = WordRed] (r2) at (-1.0, 2.25) { \color{white} \tiny $r_2$ };
		\node [rectangle, fill = WordRed] (r3) at (1.45, 2.3) { \color{white} \tiny $r_3$ };
		\node [rectangle, fill = WordRed] (rn-1) at (-0.5, -2.25) { \color{white} \tiny $r_{n \! - \! 2}$ };
		\node [rectangle, fill = WordRed] (rn-1) at (1.1, -3.75) { \color{white} \tiny $r_{n \! - \! 1}$ };
		\node [rectangle, fill = WordRed] (rn) at (2.35, -2.5) { \color{white} \tiny $r_n$ };
		\draw [GreenTeal, dashed, line width = 0.5pt] (-4.0, 2.0) -- node [above left, rotate = 45] { \tiny $diam B_1$} (-2.0, 4.0);
		\draw [GreenTeal, dashed, line width = 0.5pt] (-2.0, 4.0) -- node [above = 1.5 pt, rotate = -45] { \tiny $ diam B_2$} (0.0, 2.0);
		\draw [GreenTeal, dashed, line width = 0.5pt] (0.0, 2.0) -- node [below = 1.5 pt, rotate = 45] { \tiny $ diam B_3$} (2.0, 4.0);
		\draw [GreenTeal, dashed, line width = 0.5pt] (-2.0, -2.0) -- node [below = 1.5 pt, rotate = -45] { \tiny $ diam B_{n \! - \! 2}$} (0.0, -4.0);
		\draw [GreenTeal, dashed, line width = 0.5pt] (0.0, -4.0) -- node [above right, rotate = 45] { \tiny $ diam B_{n \! - \! 1}$} (2.0, -2.0);
		\draw [GreenTeal, dashed, line width = 0.5pt] (2.0, -4.0) -- node [below = 1.5 pt, rotate = 45] { \tiny $ diam B_n$} (4.0, -2.0);
		\draw [MyDarkBlue, line width = 1.0pt] (-3.51, 3.62) -- node [above, rotate = -27] { \footnotesize $\mathbf{d(r_1, r_2)}$} (-1.15, 2.38);
		\draw [MyDarkBlue, line width = 1.0pt] (-0.85, 2.25) -- node [above, rotate = 2] { \footnotesize $\mathbf{d(r_2, r_3)}$} (1.3, 2.3);
		\draw [MyDarkBlue, line width = 1.0pt] (-0.30, -2.38) -- node [above, rotate = -47] { \footnotesize $\mathbf{d(r_{n \! - \! 2}, r_{n \! - \! 1})}$} (0.9, -3.62);
		\draw [MyDarkBlue, line width = 1.0pt] (1.28, -3.63) -- node [below, rotate = 52] { \footnotesize $\mathbf{d(r_{n \! - \! 1}, r_n)}$} (2.21, -2.62);
	\end{tikzpicture}
	\caption{This figure shows the expected distances between restaurants located in adjacent regions.} \label{fig:Distances Between Adjucent Restaurants}
\end{figure}
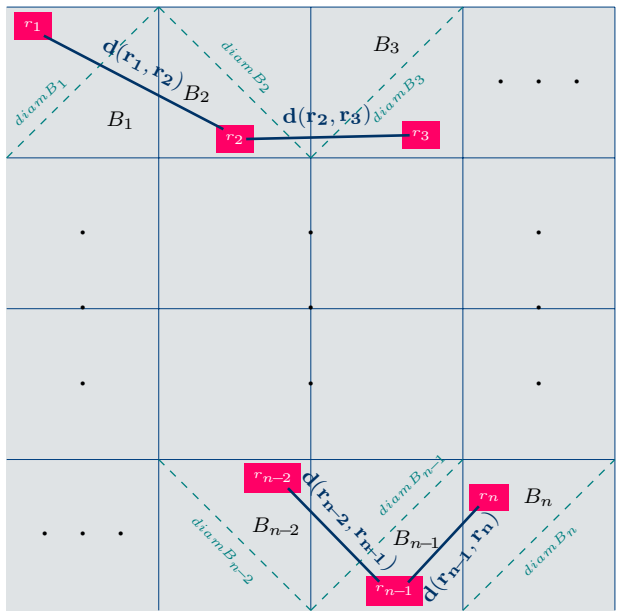

	Let us make the rather obvious observation that there is a meaningful notion of \emph{distance} defined between \emph{any two points}, or locations if you prefer, in the entire city area. In reality, this can be the geographical distance between any two locations, expressed in meters or kilometers or in some other unit of length. For instance, let us consider two points $x$ and $y$ with spatial coordinates $(x_1, x_2)$ and $(y_1, y_2)$, respectively. A typical manifestation of the notion of distance is the \emph{Euclidean} distance: $\sqrt{ (x_2 - x_1)^2 + (y_2 - y_1)^2}$ between $x$ and $y$. In any event, we take for granted the existence of such a distance function defined on every pair of points $(x, y)$ of the city, which is denoted by $d (x, y)$.

	\begin{definition} \ \label{def:Distance Adjacency Diameter Definition}
		\begin{itemize}
			\item	The distance between two regions $B_p$ and $B_q$ is defined as
			\begin{align} \label{eq:Region Distance}
				d ( B_p, B_q ) = \inf \{ d (x, y) : x \in B_p \text{ and } y \in B_q \} \ .
			\end{align}
			\item	Two regions $B_p$ and $B_q$ are \emph{adjacent} if
			\begin{align} \label{eq:Region Adjacency}
				d ( B_p, B_q ) = 0 \ .
			\end{align}
			\item	We define the concept of \emph{diameter} (see \cite{Munkres2000} for details) for the regions $B_p, 1 \leq p \leq n$. In particular, we define 
			\begin{align} \label{eq:Diameter}
				diam B_p = \sup \{ d (x, y) : x, y \in B_p \} \ .
			\end{align}
		\end{itemize}
	\end{definition}

	\begin{proposition} \label{thr: Expected Distance Between Adjacent Restaurants}
		Let the $n$ restaurants be uniformly distributed on the city area and assume that the whole area is partitioned into $n$ regions of equal area. If $r_p$ and $r_q$ are the restaurants located at adjacent regions $B_p$ and $B_q$ respectively, where $1 \leq p \neq q \leq n$, then the distance $d ( r_p, r_q )$ between them is bounded above by $diam B_p + diam B_q$:

		\begin{align} \label{eq:Upper Bound on Restaurant Distance}
			d ( r_p, r_q ) \leq diam B_p + diam B_q \ , \quad 1 \leq p \neq q \leq n \ .
		\end{align}
	\end{proposition}

	\begin{proof}
		Consider two adjacent regions $B_p$ and $B_q$. By (\ref{eq:Region Adjacency}), this means that $d ( B_p, B_q ) = 0$, which in turn implies that $\forall \varepsilon \ \exists x \in B_p \ \exists y \in B_q \text{ such that } d ( x, y ) \leq \varepsilon \quad (\star)$. In view of Proposition \ref{thr: Expected Number of Restaurants}, one expects to find exactly one restaurant in $B_p$ and exactly one restaurant in $B_q$. So, let $r_p$ and $r_q$ be the restaurants located at regions $B_p$ and $B_q$, respectively, and consider the distance $d ( r_p, r_q )$ between them. By the triangle inequality, which is a fundamental property of every distance function, we may write that $d ( r_p, r_q ) \leq d ( r_p, x ) + d ( x, y ) + d ( y, r_q ), \forall x \in B_p \ \forall y \in B_q \quad (\star\star)$. From $(\star)$ and $(\star\star)$ we conclude that $\forall \varepsilon \ \exists x \in B_p \ \exists y \in B_q \text{ such that } d ( r_p, r_q ) \leq d ( r_p, x ) + d ( y, r_q ) + \varepsilon \quad (\star\star\star)$. Now, according to (\ref{eq:Diameter}), $d ( r_p, x ) \leq diam B_p$ and $d ( y, r_q ) \leq diam B_q$. These last two relations combined with $(\star\star\star)$, give that $d ( r_p, r_q ) \leq diam B_p + diam B_q$, as desired.
	\end{proof}

	The above upper bound can be simplified if we further assume that all regions have the same geometric shape. This regularity does not impose any serious restriction on the overall setting of the game and allows us to assert that  $diam B_1 = \ldots = diam B_n = D$, in which case inequality (\ref{eq:Upper Bound on Restaurant Distance}) becomes:

	\begin{align} \label{eq:Regular Upper Bound on Restaurant Distance}
		d ( r_p, r_q ) \leq 2 D \ , \quad 1 \leq p \neq q \leq n \ .
	\end{align}

	In the special case where the regions are squares, as depicted in Figures \ref{fig:Distances Between Adjucent Restaurants} and \ref{fig:Distances Between Adjucent Restaurants Decrease}, one can easily see that the diameter $D$ is proportional to $\sqrt {\frac { 2 } { n } }$:

	\begin{align} \label{eq:Square Diameter}
		D \propto \sqrt {\frac { 2 } { n } } \ .
	\end{align}

	A comparison between Figures \ref{fig:Distances Between Adjucent Restaurants} and \ref{fig:Distances Between Adjucent Restaurants Decrease} demonstrates that the expected distance between restaurants which lie in adjacent regions is quite short, as it is bounded above by the sum of the diameters of the corresponding regions. The diameter of the regions decreases as $n$ increases, and in the special case shown in these two Figures, the diameter decreases in proportion to $\frac { 1 } { \sqrt { n } }$. In layman terms, this means that \emph{adjacent restaurants get very close to each other} as $n \to \infty$. Once the agent arrives at a restaurant, then, with high probability, visiting an adjacent restaurant will only incur a negligible extra cost that will not violate her time constraints.

\begin{figure}
	\centering
	\begin{tikzpicture} [scale = 1.2]
		\fill [WordIceBlue] (-4.0, -4.0) rectangle (4.0, 4.0);
		\draw [step = 1.0 cm, MyBlue, line width = 0.1pt] (-4.0, -4.0) grid (4.0, 4.0);
		\node [rectangle, minimum size = 2.0 cm] () at (-3.8, 3.85) { \tiny $B_1$ };
		\node [rectangle, minimum size = 2.0 cm] () at (-2.8, 3.15) { \tiny $B_2$ };
		\node [rectangle, minimum size = 2.0 cm] () at (-1.8, 3.85) { \tiny $B_3$ };
		\node [rectangle, minimum size = 2.0 cm] () at (1.3, -3.85) { \tiny $B_{n \! - \! 2}$ };
		\node [rectangle, minimum size = 2.0 cm] () at (2.3, -3.15) { \tiny $B_{n \! - \! 1}$ };
		\node [rectangle, minimum size = 2.0 cm] () at (3.8, -3.85) { \tiny $B_n$ };
		\node [rectangle] at (-0.5, 3.5){ \Large $\cdot$ };
		\node [rectangle] at (1.5, 3.5){ \Large $\cdot$ };
		\node [rectangle] at (3.5, 3.5){ \Large $\cdot$ };
		\node [rectangle] at (-3.0, 2.0){ \Large $\cdot$ };
		\node [rectangle] at (-3.0, 0.0){ \Large $\cdot$ };
		\node [rectangle] at (-3.0, -2.0){ \Large $\cdot$ };
		\node [rectangle] at (0.0, 2.0){ \Large $\cdot$ };
		\node [rectangle] at (0.0, 0.0){ \Large $\cdot$ };
		\node [rectangle] at (0.0, -2.0){ \Large $\cdot$ };
		\node [rectangle] at (3.0, 2.0){ \Large $\cdot$ };
		\node [rectangle] at (3.0, 0.0){ \Large $\cdot$ };
		\node [rectangle] at (3.0, -2.0){ \Large $\cdot$ };
		\node [rectangle] at (-3.5, -3.5){ \Large $\cdot$ };
		\node [rectangle] at (-1.5, -3.5){ \Large $\cdot$ };
		\node [rectangle] at (0.5, -3.5){ \Large $\cdot$ };
		\node [rectangle, fill = MyLightRed] (r1) at (-3.3, 3.25) { \color{gray} \tiny $r_1$ };
		\node [rectangle, fill = MyLightRed] (r2) at (-2.275, 3.78) { \color{gray} \tiny $r_2$ };
		\node [rectangle, fill = MyLightRed] (r3) at (-1.28, 3.22) { \color{gray} \tiny $r_3$ };
		\node [rectangle, fill = MyLightRed] (rn-1) at (1.41, -3.22) { \color{gray} \tiny $r_{n \! - \! 2}$ };
		\node [rectangle, fill = MyLightRed] (rn-1) at (2.4, -3.75) { \color{gray} \tiny $r_{n \! - \! 1}$ };
		\node [rectangle, fill = MyLightRed] (rn) at (3.3, -3.2) { \color{gray} \tiny $r_n$ };
		\draw [GreenTeal, dashed, line width = 0.5pt] (-4.0, 3.0) -- node [above, rotate = 45] { \tiny $diam B_1$} (-3.0, 4.0);
		\draw [GreenTeal, dashed, line width = 0.5pt] (3.0, -4.0) -- node [below, rotate = 45] { \tiny $ diam B_n$} (4.0, -3.0);
		\draw [MyDarkBlue, line width = 1.0pt] (-3.07, 3.42) -- node [above, rotate = 20] { \tiny $\mathbf{d(r_1, r_2)}$} (-2.51, 3.61);
		\draw [MyDarkBlue, line width = 1.0pt] (-2.03, 3.59) -- node [below, rotate = -15] { \tiny $\mathbf{d(r_2, r_3)}$} (-1.53, 3.4);
		\draw [MyDarkBlue, line width = 1.0pt] (1.79, -3.39) -- node [below, rotate = -35] { \tiny $\mathbf{d(r_{n \! - \! 2}, r_{n \! - \! 1})}$} (2.02, -3.55);
		\draw [MyDarkBlue, line width = 1.0pt] (2.75, -3.55) -- node [above, rotate = 35] { \tiny $\mathbf{d(r_{n \! - \! 1}, r_n)}$} (3.04, -3.36);
	\end{tikzpicture}
	\caption{When $n$ increases, the area and the diameter of the regions $B_1, \ldots, B_n$ decrease. As a result the expected distances between restaurants located in adjacent regions \emph{decrease}. In other words, as $n \to \infty$, the restaurants in adjacent regions get closer and closer.} \label{fig:Distances Between Adjucent Restaurants Decrease}
\end{figure}
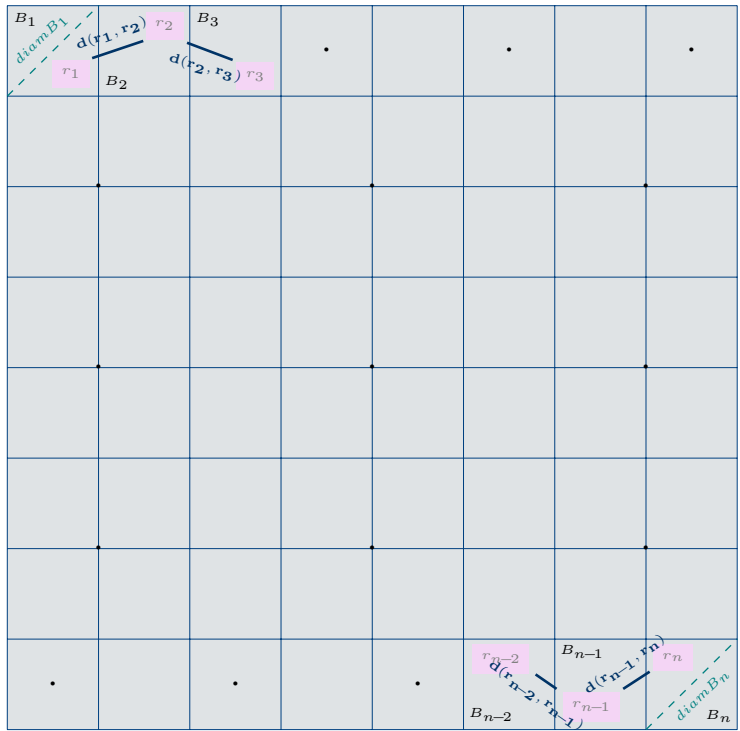

	We clarify that we are not making any assumption about the probabilistic distribution of the agents. One possibility is that the agents might be \emph{concentrated} in the ``center,'' or in another specific location of the city area, as is tacitly assumed by the original KPRP. Another possibility is that the agents follow a random distribution over the area, for instance they might also follow the uniform distribution. The former case is depicted in Figure \ref{fig:Concentrated Agents} and the latter in Figure \ref{fig:Uniformly Distributed Agents}. The crucial observation is that in both cases any of the $n$ agent can, potentially, have lunch in any of the $n$ restaurants and return back in time. This fact implies that, assuming each agent follows the (near-optimal) tour produced as a solution to her individual TSP, she may visit a second, or even a third, restaurant if her previous choices proved fruitless. To see why this is indeed so, one may consider for instance agent $a_1$ in both Figures \ref{fig:Concentrated Agents} and \ref{fig:Uniformly Distributed Agents} and the restaurant that is furthest apart. Without loss of generality let us say that in both cases this is restaurant $r_n$. Being able to visit $r_n$ while adhering to her time constraints, implies being also able to pass through adjacent restaurants within the same time window.

\begin{figure}
	\centering
	\begin{tikzpicture} [scale = 1.0]
		\fill [WordIceBlue] (-4.0, -4.0) rectangle (4.0, 4.0);
		\draw [step = 2.0 cm, MyBlue, line width = 0.25pt] (-4.0, -4.0) grid (4.0, 4.0);
		\node [rectangle, minimum size = 2.0 cm] () at (-3.0, 3.0) { $B_1$ };
		\node [rectangle, minimum size = 2.0 cm] () at (-1.0, 3.0) { $B_2$ };
		\node [rectangle, minimum size = 2.0 cm] () at (1.0, 3.0) { $B_3$ };
		\node [rectangle, minimum size = 2.0 cm] () at (-1.0, -3.0) { $B_{n-2}$ };
		\node [rectangle, minimum size = 2.0 cm] () at (1.0, -3.0) { $B_{n-1}$ };
		\node [rectangle, minimum size = 2.0 cm] () at (3.0, -3.0) { $B_n$ };
		\node [rectangle] at (2.5, 3.0){ \Large $\cdot$ };
		\node [rectangle] at (3.0, 3.0){ \Large $\cdot$ };
		\node [rectangle] at (3.5, 3.0){ \Large $\cdot$ };
		\node [rectangle] at (-3.0, 1.0){ \Large $\cdot$ };
		\node [rectangle] at (-3.0, 0.0){ \Large $\cdot$ };
		\node [rectangle] at (-3.0, -1.0){ \Large $\cdot$ };
		\node [rectangle] at (0.0, 1.0){ \Large $\cdot$ };
		\node [rectangle] at (0.0, 0.0){ \Large $\cdot$ };
		\node [rectangle] at (0.0, -1.0){ \Large $\cdot$ };
		\node [rectangle] at (3.0, 1.0){ \Large $\cdot$ };
		\node [rectangle] at (3.0, 0.0){ \Large $\cdot$ };
		\node [rectangle] at (3.0, -1.0){ \Large $\cdot$ };
		\node [rectangle] at (-3.5, -3.0){ \Large $\cdot$ };
		\node [rectangle] at (-3.0, -3.0){ \Large $\cdot$ };
		\node [rectangle] at (-2.5, -3.0){ \Large $\cdot$ };
		\node [rectangle, fill = WordRed] (r1) at (-3.5, 3.5) { \color{white} \tiny $r_1$ };
		\node [rectangle, fill = WordRed] (r2) at (-1.0, 2.5) { \color{white} \tiny $r_2$ };
		\node [rectangle, fill = WordRed] (r3) at (1.45, 2.3) { \color{white} \tiny $r_3$ };
		\node [rectangle, fill = WordRed] (rn-1) at (-1.0, -2.35) { \color{white} \tiny $r_{n-2}$ };
		\node [rectangle, fill = WordRed] (rn-1) at (0.8, -3.65) { \color{white} \tiny $r_{n-1}$ };
		\node [rectangle, fill = WordRed] (rn) at (2.5, -2.5) { \color{white} \tiny $r_n$ };
		\fill [WordAquaLighter60, line width = 1.5pt, rounded corners = 30pt](-1.15, -1.15) rectangle (1.15, 1.15);
		\node [rectangle, fill = WordBlueDarker50] (Agents) at (0.0, 4.5){ \color{white} Agents Concentrated in a Small Region};
		\node [circle, fill = WordBlueDarker50] (a1) at ( { 0.8*cos(45) }, { 0.8*sin(45) } ) { \color{white} \tiny $a_1$ };
		\node [circle, fill = WordBlueDarker50] () at ( { 0.8*cos(135) }, { 0.8*sin(135) } ) { };
		\node [circle, fill = WordBlueDarker50] () at ( { 0.8*cos(180) }, { 0.8*sin(180) } ) { };
		\node [circle, fill = WordBlueDarker50] () at ( { 0.8*cos(225) }, { 0.8*sin(225) } ) { };
		\node [circle, fill = WordBlueDarker50] (an) at ( { 0.8*cos(315) }, { 0.8*sin(315) } ) { \color{white} \tiny $a_n$ };
	\end{tikzpicture}
	\caption{The above figure depicts the situation where all $n$ agents are concentrated within a small region of the Kolkata city, while the $n$ restaurants are uniformly distributed in the overall Kolkata area.} \label{fig:Concentrated Agents}
\end{figure}
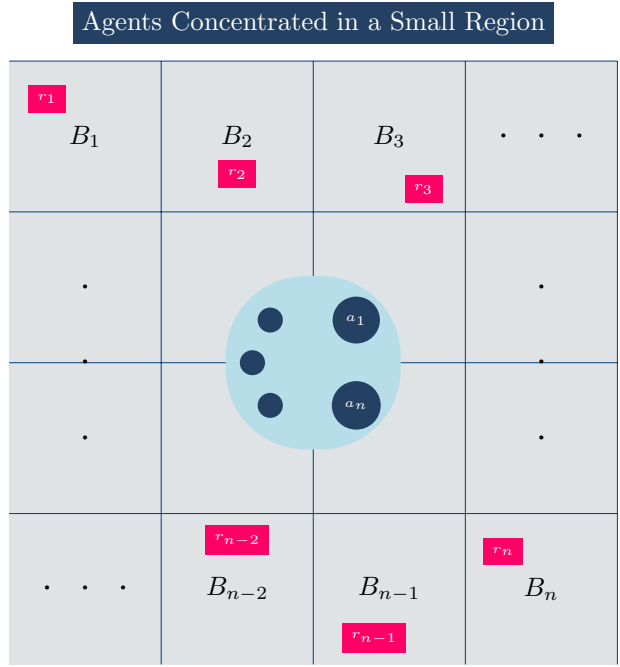

\begin{figure}
	\centering
	\begin{tikzpicture} [scale = 1.0]
		\fill [WordIceBlue] (-4.0, -4.0) rectangle (4.0, 4.0);
		\draw [step = 2.0 cm, MyBlue, line width = 0.25pt] (-4.0, -4.0) grid (4.0, 4.0);
		\node [rectangle, minimum size = 2.0 cm] () at (-3.0, 3.0) { $B_1$ };
		\node [rectangle, minimum size = 2.0 cm] () at (-1.0, 3.0) { $B_2$ };
		\node [rectangle, minimum size = 2.0 cm] () at (1.0, 3.0) { $B_3$ };
		\node [rectangle, minimum size = 2.0 cm] () at (-0.75, -3.0) { $B_{n-2}$ };
		\node [rectangle, minimum size = 2.0 cm] () at (0.75, -3.0) { $B_{n-1}$ };
		\node [rectangle, minimum size = 2.0 cm] () at (3.0, -3.0) { $B_n$ };
		\node [rectangle] at (2.5, 3.0){ \Large $\cdot$ };
		\node [rectangle] at (3.0, 3.0){ \Large $\cdot$ };
		\node [rectangle] at (3.5, 3.0){ \Large $\cdot$ };
		\node [rectangle] at (-3.0, 1.0){ \Large $\cdot$ };
		\node [rectangle] at (-3.0, 0.0){ \Large $\cdot$ };
		\node [rectangle] at (-3.0, -1.0){ \Large $\cdot$ };
		\node [rectangle] at (0.0, 1.0){ \Large $\cdot$ };
		\node [rectangle] at (0.0, 0.0){ \Large $\cdot$ };
		\node [rectangle] at (0.0, -1.0){ \Large $\cdot$ };
		\node [rectangle] at (3.0, 1.0){ \Large $\cdot$ };
		\node [rectangle] at (3.0, 0.0){ \Large $\cdot$ };
		\node [rectangle] at (3.0, -1.0){ \Large $\cdot$ };
		\node [rectangle] at (-3.5, -3.0){ \Large $\cdot$ };
		\node [rectangle] at (-3.0, -3.0){ \Large $\cdot$ };
		\node [rectangle] at (-2.5, -3.0){ \Large $\cdot$ };
		\node [rectangle, fill = WordRed] (r1) at (-3.5, 3.5) { \color{white} \tiny $r_1$ };
		\node [rectangle, fill = WordRed] (r2) at (-1.0, 2.5) { \color{white} \tiny $r_2$ };
		\node [rectangle, fill = WordRed] (r3) at (1.45, 2.3) { \color{white} \tiny $r_3$ };
		\node [rectangle, fill = WordRed] (rn-1) at (-1.0, -2.35) { \color{white} \tiny $r_{n-2}$ };
		\node [rectangle, fill = WordRed] (rn-1) at (0.8, -3.65) { \color{white} \tiny $r_{n-1}$ };
		\node [rectangle, fill = WordRed] (rn) at (2.5, -2.5) { \color{white} \tiny $r_n$ };
		\node [rectangle, fill = WordBlueDarker50] (Agents) at (0.0, 4.5){ \color{white} Agents Uniformly Distributed};
		\node [circle, fill = WordBlueDarker50] (a1) at (-2.5, 2.5) { \color{white} \tiny $a_1$ };
		\node [circle, fill = WordBlueDarker50] (a1) at (-0.5, 3.5) { \color{white} \tiny $a_2$ };
		\node [circle, fill = WordBlueDarker50] (a1) at (0.5, 2.5) { \color{white} \tiny $a_3$ };
		\node [circle, fill = WordBlueDarker50] (an) at (-1.5, -3.5) { \color{white} \tiny $a_{n-2}$ };
		\node [circle, fill = WordBlueDarker50] (an) at (1.5, -2.5) { \color{white} \tiny $a_{n-1}$ };
		\node [circle, fill = WordBlueDarker50] (an) at (3.5, -3.5) { \color{white} \tiny $a_n$ };
	\end{tikzpicture}
	\caption{This figure reflects the situation where both the $n$ restaurants and the $n$ agents are uniformly distributed in the overall Kolkata area.} \label{fig:Uniformly Distributed Agents}
\end{figure}
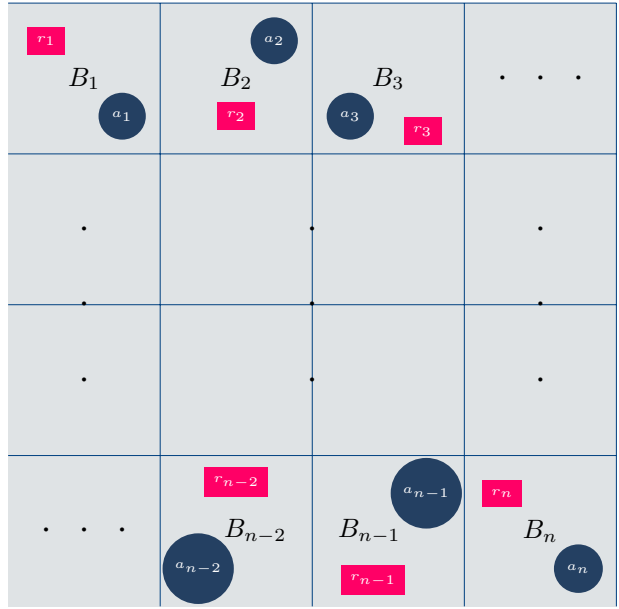

\section{Mathematical analysis of the utilization} \label{sec:Mathematical Analysis}

	The current section is devoted to the analytic estimation of the evolution of the game parameters and the daily utilization of the proposed strategy scheme. Let us briefly summarize the policy that regulates the $m$-DKPRG. 

	\begin{itemize}
		\item	At the beginning of day $1$ all $n$ agents are in the same position, in that they have not got lunch yet, and they in a precarious state not knowing if they will manage to eat eventually. So, at this point in time they are all \emph{unsatisfied}. The situation with the restaurants is symmetrical. All $n$ restaurants face uncertainty in that it is yet unknown whether they will be chosen by at least one customer. Therefore, at this point they are all \emph{vacant}.
		\item	The situation is quite different at the end of day $1$. A significant percentage of the $n$ agents, as will be shown in this section, managed to get lunch. An equal percentage of the $n$ restaurants was utilized. The common strategy followed by all agents ensures that the same agents will get lunch next day, the day after the next, etc. These agents are \emph{satisfied}, since they have effectively ``won'' the game. Symmetrically, the same restaurants will be utilized every day from now on. They will be permanently \emph{reserved}.
		\item	At the beginning of day $2$, only those agents that failed to eat yesterday will essentially play the game. These will the \emph{active} players of day $2$. The active players will strive to get lunch exclusively to the restaurants that did not serve any customer yesterday. The rest of the agents are already satisfied and will certainly have lunch today, each one at the specific restaurant that (eventually) served her yesterday.
		\item	By the end of day $2$, a significant percentage of the active agents will have succeeded in getting lunch. Thus, the total number of satisfied agents will increase by the amount of today's \emph{gains}. Of course, an equal percentage of yesterday's vacant restaurants will also be utilized for the first time today.
		\item	This process will continue ad infinitum.
	\end{itemize}

	The next concepts will prove useful in our analysis.

	\begin{definition} \
		\begin{itemize}
			\item	The expected number of agents that managed to eat lunch during day $1$ is denoted by $\overline{A_{1}^{s}}$ and the expected number of agents that failed to eat lunch during day $1$ is denoted by $\overline{A_{1}^{u}}$.
			\item	The expected number of agents that got lunch \emph{for the first time} during day $t, t = 2, 3, \ldots$, is denoted by $\overline{A_{t}^{s}}$. The expected number of agents that failed to get lunch during day $t, t = 2, 3, \ldots$, is denoted by $\overline{A_{t}^{u}}$.
			\item	Symmetrically, the expected number of restaurants that served lunch during day $1$ is denoted by $\overline{R_{1}^{r}}$ and the expected number of restaurants that did not serve lunch during day $1$ is denoted by $\overline{R_{1}^{v}}$.
			\item	The expected number of restaurants that served a customer \emph{for the first time} during day $t, t = 2, 3, \ldots$, is denoted by $\overline{R_{t}^{r}}$. The expected number of agents that failed to serve lunch during day $t, t = 2, 3, \ldots$, is denoted by $\overline{R_{t}^{v}}$.
			\item	The \emph{vacancy probability} of day $1$ is the probability that a restaurant did not accommodate any customer during day $1$ and is designated by $VP_{1}$.
			\item	The \emph{vacancy probability} of day $t, t = 2, 3, \ldots$, designated by $VP_{t}$, is the probability that a restaurant that \emph{has not served any customer before day} $t$ did not serve a customer during day $t$ either.
			\item	In the $m$-stop DKPRG, only the customers that have yet to get lunch participate actively in today's game. The agents that actually play the game at the \emph{beginning} of day $t$, seeking a restaurant to get lunch, are called \emph{active} players and their expected number is denoted by $n_t$.
			\item   The \emph{expected utilization} of day $t, t = 1, 2, \ldots$, denoted by $\overline{f_t}$, is the fraction of the expected number of agents that were served during day $t$. The \emph{steady state utilization} is defined as $f_\infty = \sup \{ f_t : t \in \mathbb{N} \}$.
		\end{itemize}
	\end{definition}

	\textbf{Equiprobability of tours.} The following analysis is based on the premise that \emph{all} $n!$ \emph{tours} are \emph{equiprobable}. In the rest of this paper we shall refer to this assumption as the \emph{equiprobability of tours assumption} (EPT for short). In view of the discussion in the previous sections, this premise is well justified.

	An immediate consequence of the EPT assumption is the \emph{equiprobability of each restaurant} appearing in any position. In particular, let us recall that in the tour $T_{a} = ( l_0, l_1, \dots, l_n, l_{n+1} )$, corresponding to agent $a$, $l_0 = l_{n+1}$ is the starting point of $a$ and $l_k, 1 \leq k \leq n$, is the index of the restaurant in the $k^{th}$ position of the tour. We may easily calculate the probability that a restaurant is in a specific position of the tour, as well as the probability of the complementary event. For easy reference, these facts are collected in the next Proposition whose proof is trivial and thus omitted.

	\begin{proposition} \label{thr: Probability of rj in position k of Ta}
		Assuming the equiprobability of tours, the following hold. 

		\begin{align} \label{eq:Probability of rj in position k of Ta}
			\forall a \in A \ \ \forall r \in R \ \ \forall k, 1 \leq k \leq n, \quad P(r \text{ is in position } k \text{ of } T_{a} ) = \frac{1}{n}
		\end{align}

		\begin{align} \label{eq:Probability of rj not in position k of Ta}
			\forall a \in A \ \ \forall r \in R \ \ \forall k, 1 \leq k \leq n, \quad P(r \ \mathrm{not} \text{ in position } k \text{ of } T_{a} ) = \frac{n - 1}{n}
		\end{align}

		The above can be generalized to handle the case of a restaurant $r$ appearing in one of $w$ distinct positions $k_1, k_2, \ldots, k_w$, where $1 < w \leq n$.

		\begin{align} \label{eq:Probability of rj in positions kw of Ta}
			\forall a \in A \ \ \forall r \in R \quad P( r \text{ is in } \mathrm{one} \text{ of positions } k_1, \ldots, k_w \text{ of } T_{a} ) = \frac{w}{n}
		\end{align}

		\begin{align} \label{eq:Probability of rj not in positions kw of Ta}
			\forall a \in A \ \ \forall r \in R \quad P( r \ \mathrm{not} \text{ in } \mathrm{any} \text{ of positions } k_1, \ldots, k_w \text{ of } T_{a} ) = \frac{n - w}{n}
		\end{align}

	\end{proposition}

	We only mention that the above hold for every restaurant, every position, and, of course, for every tour. Since the probability that restaurant $r \in R$ is in the $k^{th}$ position of the tour of agent $a$ is $\frac{1}{n}$, the probability of the complementary event, i.e., that restaurant $r$ is \emph{not} in the $k^{th}$ position of $T_{a}$ is $\frac{n - 1}{n}$. If we deem as ``success'' the case where $r$ is indeed in the $k^{th}$ position of $T_{a}$ and as ``failure'' the case where $r$ is not, then this situation is a typical example of a \emph{Bernoulli trial}, having probability of success $\frac{1}{n}$ (also referred to as \emph{parameter}, see \cite{DasGupta2010}) and probability of failure $\frac{n - 1}{n}$. In view of (\ref{eq:Probability of rj in position k of Ta}) we denote this as 

	\begin{align} \label{eq:Bernoulli Trial Parameter}
		P(r \text{ is in position } k \text{ of } T_{a}) \sim Ber(\frac{1}{n}) \ , \quad \forall a \in A \ \ \forall r \in R \ \ \forall k, 1 \leq k \leq n \ .
	\end{align}

	Analogously, the probability that restaurant $r \in R$ appears in \emph{one} of $w, 1 < w \leq n$, \emph{distinct} positions of the tour of agent $a$ is $\frac{w}{n}$. The probability of the complementary event, i.e., that restaurant $r$ is \emph{not} in any one of these $w$ positions of $T_{a}$ is $\frac{n - w}{n}$. This time, one may view as ``success'' the case where $r$ is indeed in one of the designated $w$ positions of $T_{a}$ and as ``failure'' the case where $r$ is not. So, once again we are facing with a Bernoulli trial, this time with parameter $\frac{w}{n}$. 

	\begin{align} \label{eq:General Bernoulli Trial Parameter}
		P( r \text{ is in } \mathrm{one} \text{ of positions } k_1, \ldots, k_w \text{ of } T_{a} ) \sim Ber(\frac{w}{n}) \ , \quad \forall a \in A \ \ \forall r \in R \ .
	\end{align}

	The fact that the $n$ agents calculate their tours \emph{independently}, implies that $n$ independent Bernoulli trials take place simultaneously, all with the same success and failure probabilities. This situation is described by the \emph{binomial} distribution with parameters $(n, p)$\footnote{We refer the reader to \cite{DasGupta2010} and \cite{Pitman1993} for a more detailed analysis.}, denoted by $Bin(n, p)$, where $p = \frac{1}{n}$ in the simple case of one position and $p = \frac{w}{n}$ in the general case of $w$ positions. By employing well-known formulas from probability textbooks we may assert the following Proposition, whose proof is also trivial.

	\begin{proposition} \label{thr: Probability of of l successes in n trials}
		Given a restaurant $r$, if its appearance in a specified position $k$ in one tour counts as one success, whereas its failure to appear in the specified position $k$ in one tour counts as one failure, then the probability of exactly $l$ appearances in position $k$ in total is given by
		\begin{align} \label{eq:Probability of l successes in n trials}
			\forall r \in R \ \ \forall k, 1 \leq k \leq n, \quad P( r \text{ appears } l \text{ times in position } k \text{ in } n \text{ tours} ) = \binom { n } { l } \left( \frac{1}{n} \right)^l \left( \frac{n - 1}{n} \right)^{n - l} \ .
		\end{align}
		In the special case, where $r$ \emph{never} appears, that is it appears $0$ times, in the specified position $k$, the above formula becomes:
		\begin{align} \label{eq:Probability of 0 successes in n trials}
			\forall r \in R \ \ \forall k, 1 \leq k \leq n, \quad &P( r \text{ \emph{never} appears in position } k \text{ in } n \text{ tours} ) \nonumber \\
			&= \binom { n } { 0 } \left( \frac{1}{n} \right)^0 \left( \frac{n - 1}{n} \right)^n = \left( \frac{n - 1}{n} \right)^n \ .
		\end{align}
		More generally, the probability that restaurant $r$ appears exactly $l$ times in total in one of the $w$ distinct positions $k_1, \ldots, k_w, 1 < w \leq n$ is given by
		\begin{align} \label{eq:Probability of l successes in w positions in n trials}
			\forall r \in R \quad P( r \text{ appears } l \text{ times in } \mathrm{one} \text{ of positions } k_1, \ldots, k_w \text{ in } n \text{ tours} ) = \binom { n } { l } \left( \frac{w}{n} \right)^l \left( \frac{n - w}{n} \right)^{n - l} \ .
		\end{align}
		If $r$ \emph{never} appears, that is it appears $0$ times, in anyone of the $w$ designated positions $k_1, \ldots, k_w, 1 < w \leq n$, the previous formula reduces to:
		\begin{align} \label{eq:Probability of 0 successes in w positions in n trials}
			\forall r \in R \quad &P( r \ \mathrm{never} \text{ appears in } \mathrm{any} \text{ of positions } k_1, \ldots, k_w \text{ in } n \text{ tours} ) \nonumber \\
			&= \binom { n } { 0 } \left( \frac{w}{n} \right)^0 \left( \frac{n - w}{n} \right)^n = \left( \frac{n - w}{n} \right)^n \ .
		\end{align}
	\end{proposition}

	We must emphasize that the above hold for every restaurant $r \in R$, for every position $k, 1 \leq k \leq n$, and for every set of positions $\{ k_1, \ldots, k_w \}, 1 < w \leq n$ . In other words, for every restaurant, the probability that it does not appear in one specific position in any of the $n$ tours is $\left( \frac{n - 1}{n} \right)^n$, and the probability that it does not appear in any of $w$ distinct positions in any of the $n$ tours is $\left( \frac{n - w}{n} \right)^n$.

	According to the strategy scheme employed in the $m$-stop DKPRG, at the start of the second (third, etc.) day, the satisfied customers always go straight to the restaurant that \emph{eventually} served them the previous day. We stress the word eventually because an agent may have failed to get lunch during stop $1$ of the previous day, but she may have succeeded during the second, third, or $m^{th}$ stop. This strategy is followed by all agents, something that guarantees that those customers that were satisfied on the previous day will remain satisfied today. Effectively, this strategy implies that the satisfied agents have ``won'' the game and from now on they do not need to solve their personalized TSP. The game will be played competitively by the unsatisfied agents of the previous day. We assume that they are aware of the unoccupied restaurants and, therefore, each one of them will once again solve her personalized TSP to compute her near-optimal tour. Of course, today the network of restaurants will consist of only the \emph{unoccupied} restaurants, i.e., it will be significantly smaller that yesterday. The one-shot $m$-stop DKPRG of today will be different from the one-shot game of the previous day in a critical factor: the number of ``actively competing'' players will be significantly smaller. By the nature of the game, the number of active players at the \emph{beginning} of stop $1$ of the present day is equal to the number of unsatisfied customers at the end of the previous day. The way the expected number of active players varies with each passing day is captured by the following Theorem \ref{thr: DKPRG Quantitative Characteristics}.

	\begin{theorem} \label{thr: DKPRG Quantitative Characteristics}
		The daily progression of the $m$-DKPRG is described by the following formulas, where $t$ stands for the day in question.
		\begin{align}
			VP_{t} &= \left( \frac{ n_{t} - m }{ n_{t} } \right)^{ n_{t} } \ , \ n_t \geq m, \ t = 1, 2 , \ldots \ , \label{eq:Vacancy Probability VPt} \\
			\overline{R_{t}^{v}} &= n_{t} \left( \frac{ n_{t} - m }{ n_{t} } \right)^{ n_{t} } \ , \ n_t \geq m, \ t = 1, 2 , \ldots \ , \label{eq:Expected Number of Vacant Restaurants at End of Day t} \\
			\overline{R_{t}^{r}} &= n_{t} \left( 1 - \left( \frac{ n_{t} - m }{ n_{t} } \right)^{ n_{t} } \right) \ , \ n_t \geq m, \ t = 1, 2 , \ldots \ , \label{eq:Expected Number of Reserved Restaurants at End of Day t} \\
			\overline{A_{t}^{u}} &= n_{t} \left( \frac{ n_{t} - m }{ n_{t} } \right)^{ n_{t} } \ , \ n_t \geq m, \ t = 1, 2 , \ldots \ , \label{eq:Expected Number of Unsatisfied Customers at End of Day t} \\
			\overline{A_{t}^{s}} &= n_{t} \left( 1 - \left( \frac{ n_{t} - m }{ n_{t} } \right)^{ n_{t} } \right) \ , \ n_t \geq m, \ t = 1, 2 , \ldots \ , \label{eq:Expected Number of Satisfied Customers at End of Day t} \\
			n_{1} &= n \ , \label{eq:Expected Number of Active Agents on Day 1} \\
			n_{t + 1} &= n_{t} \left( \frac{ n_{t} - m }{ n_{t} } \right)^{ n_{t} } \ , \ n_t \geq m, \ t = 2, 3, \ldots \ , \label{eq:Expected Number of Active Agents on Day t+1} \\
			\overline{f_{t}} &=
			\frac { \sum_{d = 1}^{t} n_{d} \left( 1 - \left( \frac{ n_{d} - m }{ n_{d} } \right)^{ n_{d} } \right) } { n }
			= \frac { n - n_t \left( \frac{n_t - m}{n_t} \right)^{n_t} } { n }
			\ , \ n_t \geq m, \ t = 1, 2 , \ldots \ . \label{eq:Expected Utilization on Day t}
		\end{align}
	\end{theorem}

	\begin{proof}
		The proof of the above formulas goes as follows.

		\begin{enumerate}
			\item	We first prove the auxiliary result that the vacancy probability at the \emph{beginning} of stop $z, 1 \leq z \leq m$, of day $t$ is
					\begin{align} \label{eq:Vacancy Probability VPtz}
						VP_{t, z} =  \left( \frac{ n_{t} + 1 - z }{ n_{t} } \right)^{ n_{t} } \ , \ 1 \leq z \leq m \ .
					\end{align}
					\begin{itemize}
						\item	Indeed, at the beginning of stop $1$ of day $t$, the expected number of restaurants that have not served any customer yet is equal to the expected number $n_{t}$ of active agents. On day $t$, the game is all about the active agents and the restaurants that have never been utilized up to now. At this moment in time all these restaurants are still unutilized, so vacancy is a certainty. Thus, indeed $VP_{t, 1} = 1$, which is in agreement with (\ref{eq:Vacancy Probability VPtz}) when $z = 1$.
						\item	We recall that, according to our strategy, at the beginning of day $t$ the expected number of restaurants that have not served any customer yet is equal to the expected number $n_t$ of active players. At the beginning of stop $2$ of day $t$, the probability that one of these restaurants $r$ has not served lunch yet is precisely the probability that $r$ never appears in position $1$ in any tour of the active players. This last probability is given from (\ref{eq:Probability of 0 successes in n trials}), where of course $n$ must now be replaced by $n_{t}$. Hence,
								\begin{align} \label{eq:Vacancy Probability VPt2}
									VP_{t, 2} = \left( \frac{n_{t} - 1}{n_{t}} \right)^{n_{t}} \ , \tag{ \ref{thr: DKPRG Quantitative Characteristics}.i }
								\end{align}
						which is also in agreement with (\ref{eq:Vacancy Probability VPtz}) when $z = 2$.
						\item	Let us now carefully examine what happens during stop $2$ of day $t$ of the game. According to our scheme, those customers who have failed to get lunch at their first destination will immediately proceed to their second destination. For example, if customer $a$, who follows tour $T_{a} = ( l_0, l_1, \dots, l_n, l_{n+1} )$, was not served at restaurant $r_{l_1}$, she will try restaurant $r_{l_2}$. However, an added complication arises now. It may well be the case that $r_{l_2}$ is already occupied from stop $1$. In such a case $r_{l_2}$ is completely unavailable, i.e., it is now serving another active agent. In view of this fact, we may conclude that the restaurants that are vacant at the beginning of stop $3$ must satisfy two properties:
								\begin{enumerate}
									\item[(\textbf{P1})]	they must be vacant at the beginning of stop $2$, which means that must never appear in position $1$ in any tour of the active players, and
									\item[(\textbf{P2})]	they must never appear in position $2$ in any tour of the active players.
								\end{enumerate}
								The above are summarized more succinctly in the following rule.
								\begin{enumerate}
									\item[(\textbf{C})]		The restaurants that have not served any customer up to day $t$ and are still vacant at the beginning of stop $3$ of day $t$, never appear in position $1$ or position $2$ in any tour of the active players.
								\end{enumerate}
								Therefore,
								\begin{align} \label{eq:Vacancy Probability VPt3}
									VP_{t, 3} = P( r \ \emph{never} \text{ appears in positions } 1 \text{ or } 2 \text{ in } n_t \text{ tours} )
									\overset{ (\ref{eq:Probability of 0 successes in w positions in n trials}) } { = } \left( \frac{n_t - 2}{n_t} \right)^{n_t} \ , \tag{ \ref{thr: DKPRG Quantitative Characteristics}.ii }
								\end{align}
								which is again in agreement with (\ref{eq:Vacancy Probability VPtz}) when $z = 3$.
						\item	The same reasoning can be employed to show that the vacancy probability $VP_{t, z}$ at the beginning of stop $z$ of day $t$ is
								\begin{align}
									VP_{t, z} = P( r \text{ \emph{never} appears in positions } 1, \ldots, z - 1 ) \overset{ (\ref{eq:Probability of 0 successes in w positions in n trials}) } { = } \left( \frac{n_t + 1 - z}{n_t} \right)^{n_t} \ . \tag{ \ref{thr: DKPRG Quantitative Characteristics}.iii }
								\end{align}
								Hence, we have proved the validity of (\ref{eq:Vacancy Probability VPtz}).
						\item	Finally, to calculate the probability that one of the restaurants $r$ that have not served any customer up to day $t$ is still vacant at the end of stop $m$ of day $t$, which in effect means at the end of day $t$, we must determine the probability that $r$ never appears in positions $1$, or, $2$, or $\ldots$, or $m$ in any tour of the active players. Thus,
								\begin{align}
									VP_{t} = P( r \text{ \emph{never} appears in positions } 1, \ldots, m ) \overset{ (\ref{eq:Probability of 0 successes in w positions in n trials}) } { = } \left( \frac{n_t - m}{n_t} \right)^{n_t} \ , \tag{ \ref{thr: DKPRG Quantitative Characteristics}.iv }
								\end{align}
								which verifies (\ref{eq:Vacancy Probability VPt}), as desired. However, there is one final detail that we must emphasize here. Probabilities are real numbers taking values in the real line interval $[0, 1]$. For the equation (\ref{eq:Vacancy Probability VPt}) to be valid, it must hold that $n_t \geq m$, otherwise it cannot be regarded as a probability. The physical meaning of this restriction is that (\ref{eq:Vacancy Probability VPt}) is meaningful and correct as long as there are at least as many active players as stops $m$. If on some day $t$ we have that $n_t \leq m$, then the strategy we adhere to will make sure that all $n_t$ active players will manage to get lunch during day $t$.
					\end{itemize}
			\item	Let us clarify that our sample space consists precisely of the restaurants that have not served any customer up to day $t$. The expected number of restaurants in our sample space that remained vacant at the \emph{end} of day $t$ is given by $\overline{R_{t}^{v}}$. First, we express probabilistically those restaurants of our sample space that remain vacant after all agents visit their first $m$ destinations. We define the family of random variables $R_{t j}^{v}, \ 1 \leq j \leq n_t$. The random variable $R_{t j}^{v}$ indicates whether restaurant $r_j$ is vacant or not at the end of day $t$. Specifically, if $R_{t j}^{v}$ has the value $1$, then restaurant $r_j$ is \emph{vacant} at the end of day $t$, whereas if $R_{t j}^{v}$ is $0$, then $r_j$ is occupied.
					\begin{align} \label{eq:Random Variables Rtjv}
						R_{t j}^{v} =
						\left\{
						\begin{matrix*}[l]
							1 & \text{if restaurant } r_j \text{ is \emph{vacant} at the end of day } t \\
							0 & \text{otherwise}
						\end{matrix*} 
						\right.
						\ , \quad 1 \leq j \leq n_t \ . \tag{ \ref{thr: DKPRG Quantitative Characteristics}.v }
					\end{align}
					By combining definition (\ref{eq:Random Variables Rtjv}) and equation (\ref{eq:Vacancy Probability VPt}), we deduce that
					\begin{align} \label{eq:Probability of Random Variables Rtjv}
						R_{t j}^{v} =
						\left\{
						\begin{matrix*}[l]
							1 & \text{with probability } \left( \frac{n_t - m}{n_t} \right)^{n_t} \\
							0 & \text{with probability } 1 - \left( \frac{n_t - m}{n_t} \right)^{n_t}
						\end{matrix*}
						\right.
						\ , \quad 1 \leq j \leq n_t \ . \tag{ \ref{thr: DKPRG Quantitative Characteristics}.vi }
					\end{align}
					Having done that, we define the random variable $R_{t}^{v}$, which counts the \emph{the number of restaurants that are vacant} at the end of day $t$.
					\begin{align} \label{eq:Random Variable Rtv}
						R_{t}^{v} = \sum_{ j = 1 }^{ n_t } R_{t j}^{v} \ . \tag{ \ref{thr: DKPRG Quantitative Characteristics}.vii }
					\end{align}
					As always, in this probabilistic setting, we are interested not in the actual value of the random variable $R_{t}^{v}$, but in its expected value $E [ R_{t}^{v} ]$. In view of definition (\ref{eq:Random Variable Rtv}) and the linearity of the expected value operator, we derive that 
					\begin{align} \label{eq:Expected Value of Random Variable R_{t}^{v}}
						\overline{R_{t}^{v}} &= E \left[ R_{t}^{v} \right]
						\overset{ (\ref{eq:Random Variable Rtv}) } { = }
						E \left[ \sum_{ j = 1 }^{ n_t } R_{t j}^{v} \right] =
						\sum_{ j = 1 }^{ n_t } E \left[ R_{t j}^{v} \right] \nonumber \\
						&\overset{ (\ref{eq:Probability of Random Variables Rtjv}) } { = }
						\sum_{ j = 1 }^{ n_t } \left( 1 \cdot \left( \frac{n_t - m}{n_t} \right)^{n_t} \right) = n_t \left( \frac{n_t - m}{n_t} \right)^{n_t}
						\ , \tag{ \ref{thr: DKPRG Quantitative Characteristics}.viii }
					\end{align}
					which verifies (\ref{eq:Expected Number of Vacant Restaurants at End of Day t}).
			\item	Recall that our sample space contains exactly those restaurants that have not served any customer up to day $t$. $\overline{R_{t}^{r}}$ denotes the expected number of the restaurants of the sample space that were visited by an agent by the \emph{end} of day $t$. Now, we define the family of random variables $R_{t j}^{r}, \ 1 \leq j \leq n_t$, which indicate whether restaurant $r_j$ is occupied or not at the end of day $t$. Specifically, if $R_{t j}^{r}$ has the value $1$, then restaurant $r_j$ is \emph{occupied} at the end of day $t$, whereas if $R_{t j}^{r}$ is $0$, then $r_j$ is vacant.
					\begin{align} \label{eq:Random Variables Rtjr}
						R_{t j}^{r} =
						\left\{
						\begin{matrix*}[l]
							1 & \text{if restaurant } r_j \text{ is \emph{occupied} at the end of day } t \\
							0 & \text{otherwise}
						\end{matrix*} 
						\right.
						\ , \quad 1 \leq j \leq n_t \ . \tag{ \ref{thr: DKPRG Quantitative Characteristics}.ix }
					\end{align}
					By combining definition (\ref{eq:Random Variables Rtjr}) and equation (\ref{eq:Vacancy Probability VPt}), we deduce that
					\begin{align} \label{eq:Probability of Random Variables Rtjr}
						R_{t j}^{r} =
						\left\{
						\begin{matrix*}[l]
							1 & \text{with probability } 1 - \left( \frac{n_t - m}{n_t} \right)^{n_t} \\
							0 & \text{with probability } \left( \frac{n_t - m}{n_t} \right)^{n_t}
						\end{matrix*}
						\right.
						\ , \quad 1 \leq j \leq n_t \ . \tag{ \ref{thr: DKPRG Quantitative Characteristics}.x }
					\end{align}
					Having done that, we define the random variable $R_{t}^{r}$, which counts the \emph{the number of restaurants that are occupied} at the end of day $t$.
					\begin{align} \label{eq:Random Variable Rtr}
						R_{t}^{r} = \sum_{ j = 1 }^{ n_t } R_{t j}^{r} \ . \tag{ \ref{thr: DKPRG Quantitative Characteristics}.xi }
					\end{align}
					We are not interested in the actual value of the random variable $R_{t}^{r}$, but in its expected value $E [ R_{t}^{r} ]$. In view of definition (\ref{eq:Random Variable Rtr}) and the linearity of the expected value operator, we derive that
					\begin{align} \label{eq:Expected Value of Random Variable R_{t}^{r}}
						\overline{R_{t}^{r}} &= E \left[ R_{t}^{r} \right] 
						\overset{ (\ref{eq:Random Variable Rtr}) } { = }
						E \left[ \sum_{ j = 1 }^{ n_t } R_{t j}^{r} \right] =
						\sum_{ j = 1 }^{ n_t } E \left[ R_{t j}^{r} \right] \nonumber \\
						&\overset{ (\ref{eq:Probability of Random Variables Rtjr}) } { = }
						\sum_{ j = 1 }^{ n_t } 1 \cdot \left( 1 - \left( \frac{n_t - m}{n_t} \right)^{n_t} \right) = n_t \left( 1 - \left( \frac{n_t - m}{n_t} \right)^{n_t} \right)
						\ , \tag{ \ref{thr: DKPRG Quantitative Characteristics}.xii }
					\end{align}
					which verifies (\ref{eq:Expected Number of Reserved Restaurants at End of Day t}).
			\item	The rules of the game stipulate that the number of customers that have not managed to eat lunch at the end of day $t$ is equal to the number of restaurants that have not served any customer at the end of day $t$. Hence, their expected values are also equal, which means that $\overline{A_{t}^{u}} = \overline{R_{t}^{v}}$ and (\ref{eq:Expected Number of Unsatisfied Customers at End of Day t}) is proved.
			\item	Likewise, the adopted strategy ensures that the number of the active players that succeeded in getting lunch at the \emph{end} of day $t$ is equal to the number of restaurants that, although they had not served any agent up to day $t$, they managed to accommodate a customer by the end of day $t$. Hence, their expected values are also equal, which means that $\overline{A_{t}^{s}} = \overline{R_{t}^{r}}$ and (\ref{eq:Expected Number of Satisfied Customers at End of Day t}) is proved.
			\item	We are now in a position that enables us to assert the expected number of active players.
					\begin{itemize}
						\item	At the beginning of the first day, the numbers of active players is exactly $n$. This trivial observation confirms the initial condition (\ref{eq:Expected Number of Active Agents on Day 1}).
						\item	As we have previously explained, the adopted strategy in the $m$-DKPRG ensures that the number of agents that have not got lunch at the end of day $t$ is always equal to the number of active players on day $t + 1$. Thus, the expected number of active agents on day $t + 1$ is equal to the expected number of unsatisfied agents at the end of day $t$:
								\begin{align} \label{eq:Expected Number of Active Agents Equal to Unsatisfied Agents}
									n_{t + 1} = \overline{A_{t}^{u}} \ . \tag{ \ref{thr: DKPRG Quantitative Characteristics}.xiii }
								\end{align}
								By combining the previous result (\ref{eq:Expected Number of Active Agents Equal to Unsatisfied Agents}) with (\ref{eq:Expected Number of Unsatisfied Customers at End of Day t}), we derive
								\begin{align} \label{Expected Number of Active Agents}
									n_{t + 1} = n_{t} \left( \frac{ n_{t} - m }{ n_{t} } \right)^{ n_{t} } \ , \tag{ \ref{thr: DKPRG Quantitative Characteristics}.xiv }
								\end{align}
								which establishes the validity of (\ref{eq:Expected Number of Active Agents on Day t+1}), as desired.
					\end{itemize}
			\item	The expected utilization $\overline{f_t}$ for day $t = 1, 2 , \ldots$, is the ratio of the expected number of agents that were served during day $t$. This last numbers is equal to the expected number of customers that got lunch on day $1$, plus the expected number of the \emph{additional} customers that got lunch on day $2$, and so on. The \emph{additional} agents of day $t$ are precisely those agents that had failed to get lunch prior to day $t$, but succeeded in eating on day $t$. Their expected number is $\overline{A_{t}^{s}}$, which is given by equation (\ref{eq:Expected Number of Satisfied Customers at End of Day t}). Hence, the total number of agents that have eaten lunch up to and including day $t$ is given by
					\begin{align} \label{eq:Total Expected Number of Satisfied Agents on Day t}
						\sum_{d = 1}^{t} A_{d}^{s}
						\overset{ (\ref{eq:Expected Number of Satisfied Customers at End of Day t}) } { = }
						\sum_{d = 1}^{t} n_{d} \left( 1 - \left( \frac{ n_{d} - m }{ n_{d} } \right)^{ n_{d} } \right)
						\ , \ t = 1, 2 , \ldots \ . \tag{ \ref{thr: DKPRG Quantitative Characteristics}.xv }
					\end{align}
					An equivalent way to compute this exact number is by subtracting from the total number of agents $n$ the expected number of agents that failed to get lunch on day $t$, which is $\overline{A_{t}^{u}}$, which is given by equation (\ref{eq:Expected Number of Unsatisfied Customers at End of Day t}). Thus,
					\begin{align} \label{eq:Expected Number of Remaing Unsatisfied Agents on Day t}
						n - \overline{A_{t}^{u}} 
						\overset{ (\ref{eq:Expected Number of Unsatisfied Customers at End of Day t}) } { = }
						n - n_t \left( \frac{n_t - m}{n_t} \right)^{n_t}
						\ , \ t = 1, 2 , \ldots \ . \tag{ \ref{thr: DKPRG Quantitative Characteristics}.xvi }
					\end{align}
					Together the two formulas (\ref{eq:Total Expected Number of Satisfied Agents on Day t}) and (\ref{eq:Expected Number of Remaing Unsatisfied Agents on Day t}), allow us to conclude that
					\begin{align} \label{Expected Utilization}
						\overline{f_{t}}
						\overset{ (\ref{eq:Total Expected Number of Satisfied Agents on Day t}) } { = }
						\frac { \sum_{d = 1}^{t} n_{d} \left( 1 - \left( \frac{ n_{d} - m }{ n_{d} } \right)^{ n_{d} } \right) } { n }
						\overset{ (\ref{eq:Expected Number of Remaing Unsatisfied Agents on Day t}) } { = }
						\frac { n - n_t \left( \frac{n_t - m}{n_t} \right)^{n_t} } { n }
						\ , \ t = 1, 2 , \ldots \ , \tag{ \ref{thr: DKPRG Quantitative Characteristics}.xvii }
					\end{align}
					which establishes the validity of (\ref{eq:Expected Utilization on Day t}), as desired.
		\end{enumerate}
	\end{proof}

	Let us now make an important observation: formula (\ref{eq:Expected Number of Vacant Restaurants at End of Day t}) that we derived above, and which gives the expected number of vacant restaurants at the end of day $t$, is completely general and subsumes more special formulas found in the literature. Take for example the special case where $t = 1$ and $m = 1$. For these values, (\ref{eq:Expected Number of Vacant Restaurants at End of Day t}) computes the expected number of vacant restaurants at the end of day $1$ for the standard one-stop KPRP. By subtracting this quantity from $n$, the number of initially available restaurants, and then dividing by $n$, we derive the expected utilization ratio for day $1$. Indeed

	\begin{align} \label{eq:General Form Day 1 Utilization}
		\overline{f_1} = \frac { n - n \left( \frac{n - 1}{n} \right)^n } { n } = 1 - \left( \frac{n - 1}{n} \right)^n \ .
	\end{align}

	One assumption that is taken for granted in the literature is that the number of agents $n$ tends to infinity. It is straightforward to see how the above formula simplifies when $n \to \infty$. We recall a very useful fact from calculus (see for instance \cite{RobertAdams2016}), namely that
	\begin{align} \label{eq:The function Expx as a Limit}
		\forall x \in \mathbb{R}, \quad \lim_{n \to \infty} \left( 1 + \frac{x}{n}\right)^{n} \to e^{x} \ .
	\end{align}
	Under this premise, we see that $\lim_{n \to \infty} f \to 1 - e^{-1}$, which is in complete agreement with a well-known result of the literature.

	\begin{corollary} \label{thr: Approximation of DKPRG Quantitative Characteristics}
		If we assume that $n \to \infty$, then the following approximations hold, where $t$ is the day in question.
		\begin{align}
			n_{t + 1} &\approx n e^{-t m} \ , \ t = 1, 2 , \ldots \ , \label{eq:Approximation of Expected Number of Active Agents on Day t+1} \\
			\overline{A_{t}^{u}} &\approx n e^{-t m} \ , \ t = 1, 2 , \ldots \ , \label{eq:Approximation of Expected Number of Unsatisfied Customers on Day t} \\
			\overline{A_{t}^{s}} &\approx n \left( 1 - e^{ - m } \right) e^{ - ( t - 1 ) m } 
			= n e^{ - ( t - 1 ) m } - n e^{ -t m } \ , \ t = 1, 2 , \ldots \ , \label{eq:Approximation of Expected Number of Satisfied Agents on Day t} \\
			\overline{R_{t}^{v}} &\approx n e^{-t m} \ , \ t = 1, 2 , \ldots \ , \label{eq:Approximation of Expected Number of Vacant Restaurants on Day t} \\
			\overline{R_{t}^{r}} &\approx n e^{ -(t - 1) m} - n e^{ -t m} \ , \ t = 1, 2 , \ldots \ , \label{eq:Approximation of Expected Number of Reserved Restaurants on Day t} \\
			VP_{t} &\approx e^{-t m} \ , \ t = 1, 2 , \ldots \ , \label{eq:Approximation of Vacancy Probability on Day t} \\
			\overline{f_{t}} &\approx 1 - e^{ -t m } \ , \ t = 1, 2 , \ldots \ . \label{eq:Approximate Expected Utilization on Day t} \\
			\overline{f_{\infty}} &\approx 1 \ . \label{eq:Approximate Steady State Utilization}
		\end{align}
	\end{corollary}

	\begin{proof}
		The above approximations are easily proved as shown below.

		\begin{enumerate}
			\item	If we invoke property (\ref{eq:The function Expx as a Limit}), we deduce that $\left( \frac{n - m}{n} \right)^n \xrightarrow [n \to \infty] {} e^{-m}$ and $\left( \frac{ n_{t} - m }{ n_{t} } \right)^{ n_{t} } \xrightarrow [n_t \to \infty] {} e^{-m}$. When $n$ and $n_t$ do not take very large values, these limits can only serve as good approximations. Thus, it is more accurate to write  
					\begin{align} \label{eq:Exponential Approximation}
						\left( \frac{n - m}{n} \right)^n \approx e^{-m}
						\quad \text{ and } \quad
						\left( \frac{ n_{t} - m }{ n_{t} } \right)^{ n_{t} } \approx e^{-m} \ . \tag{ \ref{thr: Approximation of DKPRG Quantitative Characteristics}.i }
					\end{align}
					Formulas (\ref{eq:Expected Number of Active Agents on Day t+1}) and (\ref{eq:Expected Number of Active Agents on Day 1}) imply that
					\begin{align} \label{eq:Approximate Number of Active Agents on Day 2}
						n_2 =  n \left( \frac{n - m}{n} \right)^n
						\overset{ (\ref{eq:Exponential Approximation}) } { \approx }
						n e^{-m} \ . \tag{ \ref{thr: Approximation of DKPRG Quantitative Characteristics}.ii }
					\end{align}
					In an identical manner, we see that
					\begin{align} \label{eq:Approximate Number of Active Agents on Day 3}
						n_3 =  n_2 \left( \frac{n_2 - m}{n_2} \right)^{n_2}
						\overset{ (\ref{eq:Exponential Approximation}) } { \approx } n_2 e^{-m}
						\overset{ (\ref{eq:Approximate Number of Active Agents on Day 2}) } { \approx } n e^{-m} e^{-m} = n e^{-2 m} \ . \tag{ \ref{thr: Approximation of DKPRG Quantitative Characteristics}.iii }
					\end{align}
					Following this line of thought, it is now routine to see that (\ref{eq:Approximation of Expected Number of Active Agents on Day t+1}) holds.
			\item	The proof is trivial because the number of customers that have not eaten lunch by the end of day $t$ is equal to the number of active agents at the beginning of day $t + 1$.
			\item
					Considering the fact that $\left( \frac{ n_{t} - m }{ n_{t} } \right)^{ n_{t} } \xrightarrow [n_t \to \infty] {} e^{-m}$, we may deduce that $1 - \left( \frac{ n_{t} - m }{ n_{t} } \right)^{ n_{t} } \xrightarrow [n_t \to \infty] {} 1 - e^{-m}$. In this way we have derived the next approximation.
					\begin{align} \label{eq:Approximate Rate of success}
						1 - \left( \frac{ n_{t} - m }{ n_{t} } \right)^{ n_{t} } \approx 1 - e^{-m} \ . \tag{ \ref{thr: Approximation of DKPRG Quantitative Characteristics}.iv }
					\end{align}
					It follows from (\ref{eq:Approximation of Expected Number of Active Agents on Day t+1}) that
					\begin{align} \label{eq:Approximation of Expected Number of Active Agents on Day t}
						n_{t} \approx n e^{ -(t - 1) m} \ . \tag{ \ref{thr: Approximation of DKPRG Quantitative Characteristics}.v }
					\end{align}
					Together (\ref{eq:Approximate Rate of success}) and (\ref{eq:Approximation of Expected Number of Active Agents on Day t}) imply that 
					\begin{align} \label{eq:Approximate Expected Number of Ssatisfied Customers}
						\overline{A_{t}^{s}} &= n_{t} \left( 1 - \left( \frac{ n_{t} - m }{ n_{t} } \right)^{ n_{t} } \right)
						\overset{ (\ref{eq:Approximation of Expected Number of Active Agents on Day t}), (\ref{eq:Approximate Rate of success}) } { \approx }
						n \left( 1 - e^{ - m } \right) e^{ - ( t - 1 ) m } \nonumber \\
						&= n e^{ - ( t - 1 ) m } - n e^{ - ( t - 1 ) m } e^{ - m }
						= n e^{ - ( t - 1 ) m } - n e^{ -t m }
						\ , \tag{ \ref{thr: Approximation of DKPRG Quantitative Characteristics}.vi }
					\end{align}
					which proves equation (\ref{eq:Expected Utilization on Day t}), as desired.
			\item	Again, this result is obvious because the number of vacant restaurants at the end of day $t$ is equal to the number of unsatisfied customers at the end of day $t$.
			\item	The number of restaurants that served a customer for the first time during day $t$ is equal to the number of agents that managed to get lunch for the first time during day $t$. Hence, the result is immediate.
			\item	This can be easily shown by substituting in formula (\ref{eq:Vacancy Probability VPt}) the approximation (\ref{eq:Exponential Approximation}).
			\item	A good approximation for the \emph{total} number of restaurants that were utilized during day $t$ can be found by subtracting from $n$ the approximate expected number of customers that did not have lunch on day $t$. In view of (\ref{eq:Approximation of Expected Number of Vacant Restaurants on Day t}), this implies that
					\begin{align} \label{eq:Approximate Expected Utilization I}
						\overline{f_{t}}
						\overset{ (\ref{eq:Approximation of Expected Number of Vacant Restaurants on Day t}) } { \approx }
						\frac { n - n e^{ -t m } } { n }
						= 1 - e^{ -t m }
						\ , \ t = 1, 2 , \ldots \ . \tag{ \ref{thr: Approximation of DKPRG Quantitative Characteristics}.vii }
					\end{align}
					Alternatively, one may use the expected number of agents that ate lunch on day $t$, which can be computed as the sum $\sum_{d = 1}^{t} \overline{A_{d}^{s}}$. Using the relation (\ref{eq:Approximation of Expected Number of Satisfied Agents on Day t}), this sum can be written as
					\begin{align} \label{eq:Approximate Sum of Satisfied Customers}
						\sum_{d = 1}^{t} n \left( 1 - e^{ - m } \right) e^{ - ( d - 1 ) m }
						&=  n \left( 1 - e^{ - m } \right) \sum_{d = 1}^{t}e^{ - ( d - 1 ) m } \nonumber \\
						&=  n \left( 1 - e^{ - m } \right) \frac { 1 - e^{ - t m } } { 1 - e^{ m } }
						=  n \left(  1 - e^{ - t m } \right)
						\ , \ t = 1, 2 , \ldots \ . \tag{ \ref{thr: Approximation of DKPRG Quantitative Characteristics}.viii }
					\end{align}
					In the above derivation, we used a well-know fact (see for instance \cite{Axler2013}), namely that the sum of the first $t$ terms of a geometric sequence with first term $g_1$ and ration $\rho$ is given by the formula $g_1 + g_1 \rho + g_1 \rho^2 + \dots + g_1 \rho^{t - 1} = g_1 \frac { 1 - \rho^t } { 1 - \rho }$. This second way to estimate the expected utilization confirms, as expected, that
					\begin{align} \label{eq:Approximate Expected Utilization II}
						\overline{f_{t}} = \sum_{d = 1}^{t} \overline{A_{d}^{s}}
						\overset{ (\ref{eq:Approximate Sum of Satisfied Customers}) } { \approx }
						\frac { n \left(  1 - e^{ - t m } \right) } { n }
						= 1 - e^{ -t m }
						\ , \ t = 1, 2 , \ldots \ , \tag{ \ref{thr: Approximation of DKPRG Quantitative Characteristics}.ix }
					\end{align}
					which proves equation (\ref{eq:Approximate Expected Utilization on Day t}), as desired.
			\item	A trivial consequence of (\ref{eq:Approximate Expected Utilization on Day t}), as $t \to \infty$.
		\end{enumerate}
	\end{proof}

	To demonstrate how the exact formulas (\ref{eq:Vacancy Probability VPt}) - (\ref{eq:Expected Utilization on Day t}) reflect the daily evolution of the $m$-DKPRG we study five typical instances of the game. The first four are instances of $2$-DKPRG games with substantially different number of players. In the first four games, the number of steps $m$ is $2$, meaning that each agent may visit two restaurants if the need arises. In the first example the number of agents $n$ is $100$, a relatively small number, and its detailed progression is shown in Table \ref{tbl:The progression of the mDKPRG for m 2 and n 100}. The steady state utilization is, as expected, $1$ and it is achieved by the end of day $3$.

\begin{table}[H]
	\centering
	\STautoround{5}
	\begin{spreadtab}{{tabular}{cc||ccccc}}
		\hline
		\multicolumn{2}{c||}{@ Number of stops $m$} & \multicolumn{5}{c}{:={2}} \\
		\multicolumn{2}{c||}{@ Number of agents $n$} & \multicolumn{5}{c}{:={100}} \\
		\multicolumn{2}{c||}{@ Day ($t$)}
		& @ $n_t$ 
		& @ $VP_{t}$ 
		& @ $\overline{A_{t}^{s}} = \overline{R_{t}^{r}}$
		& @ $\overline{A_{t}^{u}} = \overline{R_{t}^{v}}$
		& @ $\overline{f_{t}}$ \\
		\hline
		@ {Start of day} & 1 & c2 & 1 & 0 & c2 & 0 \\
		@ {End of day} & 1
		& ifgt([0,-1], c1, [0,-1] * ( ( [0,-1] - c1 ) / [0,-1] )^[0,-1], 0)
		& ifgt([-1,-1], c1, ( ( [-1,-1] - c1 ) / [-1,-1] )^[-1,-1], 0)
		& ifgt([-2,-1], c1, [-2,-1] - ( [-2,-1] * [-1,0] ), [-2,-1])
		& ifgt([-3,-1], c1, [-3,-1] * [-2,0], 0)
		& ifgt([-4,-1], c1, ( c2 - [-1,0] ) / c2, 1) \\
		\hline
		@ {Start of day} & [0,-2] + 1
		& [0,-1]
		& 1
		& 0
		& [0,-1]
		& [0,-1] \\
		@ {End of day} & [0,-2] + 1
		& ifgt([0,-1], c1, [0,-1] * ( ( [0,-1] - c1 ) / [0,-1] )^[0,-1], 0)
		& ifgt([-1,-1], c1, ( ( [-1,-1] - c1 ) / [-1,-1] )^[-1,-1], 0)
		& ifgt([-2,-1], c1, [-2,-1] - ( [-2,-1] * [-1,0] ), [-2,-1])
		& ifgt([-3,-1], c1, [-3,-1] * [-2,0], 0)
		& ifgt([-4,-1], c1, ( c2 - [-1,0] ) / c2, 1) \\
		\hline
		@ {Start of day} & [0,-2] + 1 
		& [0,-1]
		& 1
		& 0
		& [0,-1]
		& [0,-1] \\
		@ {End of day} & [0,-2] + 1
		& ifgt([0,-1], c1, [0,-1] * ( ( [0,-1] - c1 ) / [0,-1] )^[0,-1], 0)
		& ifgt([-1,-1], c1, ( ( [-1,-1] - c1 ) / [-1,-1] )^[-1,-1], 0)
		& ifgt([-2,-1], c1, [-2,-1] - ( [-2,-1] * [-1,0] ), [-2,-1])
		& ifgt([-3,-1], c1, [-3,-1] * [-2,0], 0)
		& ifgt([-4,-1], c1, ( c2 - [-1,0] ) / c2, 1) \\
		\hline
		\hline
	\end{spreadtab}
	\caption{This Table demonstrates the progression of the $m$-DKPRG for $m = 2$ and $n = 100$. It can be seen that all restaurants are utilized by the end of day $3$.} \label{tbl:The progression of the mDKPRG for m 2 and n 100}
\end{table}

	In the second example the number of agents $n$ is $1000$ and its progression is shown in Table \ref{tbl:The progression of the mDKPRG for m 2 and n 1000}. The steady state utilization is, as expected, $1$ and it is achieved by the end of day $5$, i.e., $2$ days later compared to the previous example.

\begin{table}[H]
	\centering
	\STautoround{5}
	\begin{spreadtab}{{tabular}{cc||ccccc}}
		\hline
		\multicolumn{2}{c||}{@ Number of stops $m$} & \multicolumn{5}{c}{:={2}} \\
		\multicolumn{2}{c||}{@ Number of agents $n$} & \multicolumn{5}{c}{:={1000}} \\
		\multicolumn{2}{c||}{@ Day ($t$)}
		& @ $n_t$ 
		& @ $VP_{t}$ 
		& @ $\overline{A_{t}^{s}} = \overline{R_{t}^{r}}$
		& @ $\overline{A_{t}^{u}} = \overline{R_{t}^{v}}$
		& @ $\overline{f_{t}}$ \\
		\hline
		@ {Start of day} & 1 & c2 & 1 & 0 & c2 & 0 \\
		@ {End of day} & 1
		& ifgt([0,-1], c1, [0,-1] * ( ( [0,-1] - c1 ) / [0,-1] )^[0,-1], 0)
		& ifgt([-1,-1], c1, ( ( [-1,-1] - c1 ) / [-1,-1] )^[-1,-1], 0)
		& ifgt([-2,-1], c1, [-2,-1] - ( [-2,-1] * [-1,0] ), [-2,-1])
		& ifgt([-3,-1], c1, [-3,-1] * [-2,0], 0)
		& ifgt([-4,-1], c1, ( c2 - [-1,0] ) / c2, 1) \\
		\hline
		@ {Start of day} & [0,-2] + 1
		& [0,-1]
		& 1
		& 0
		& [0,-1]
		& [0,-1] \\
		@ {End of day} & [0,-2] + 1
		& ifgt([0,-1], c1, [0,-1] * ( ( [0,-1] - c1 ) / [0,-1] )^[0,-1], 0)
		& ifgt([-1,-1], c1, ( ( [-1,-1] - c1 ) / [-1,-1] )^[-1,-1], 0)
		& ifgt([-2,-1], c1, [-2,-1] - ( [-2,-1] * [-1,0] ), [-2,-1])
		& ifgt([-3,-1], c1, [-3,-1] * [-2,0], 0)
		& ifgt([-4,-1], c1, ( c2 - [-1,0] ) / c2, 1) \\
		\hline
		@ {Start of day} & [0,-2] + 1
		& [0,-1]
		& 1
		& 0
		& [0,-1]
		& [0,-1] \\
		@ {End of day} & [0,-2] + 1
		& ifgt([0,-1], c1, [0,-1] * ( ( [0,-1] - c1 ) / [0,-1] )^[0,-1], 0)
		& ifgt([-1,-1], c1, ( ( [-1,-1] - c1 ) / [-1,-1] )^[-1,-1], 0)
		& ifgt([-2,-1], c1, [-2,-1] - ( [-2,-1] * [-1,0] ), [-2,-1])
		& ifgt([-3,-1], c1, [-3,-1] * [-2,0], 0)
		& ifgt([-4,-1], c1, ( c2 - [-1,0] ) / c2, 1) \\
		\hline
		@ {Start of day} & [0,-2] + 1
		& [0,-1]
		& 1
		& 0
		& [0,-1]
		& [0,-1] \\
		@ {End of day} & [0,-2] + 1
		& ifgt([0,-1], c1, [0,-1] * ( ( [0,-1] - c1 ) / [0,-1] )^[0,-1], 0)
		& ifgt([-1,-1], c1, ( ( [-1,-1] - c1 ) / [-1,-1] )^[-1,-1], 0)
		& ifgt([-2,-1], c1, [-2,-1] - ( [-2,-1] * [-1,0] ), [-2,-1])
		& ifgt([-3,-1], c1, [-3,-1] * [-2,0], 0)
		& ifgt([-4,-1], c1, ( c2 - [-1,0] ) / c2, 1) \\
		\hline
		@ {Start of day} & [0,-2] + 1 
		& [0,-1]
		& 1
		& 0
		& [0,-1]
		& [0,-1] \\
		@ {End of day} & [0,-2] + 1
		& ifgt([0,-1], c1, [0,-1] * ( ( [0,-1] - c1 ) / [0,-1] )^[0,-1], 0)
		& ifgt([-1,-1], c1, ( ( [-1,-1] - c1 ) / [-1,-1] )^[-1,-1], 0)
		& ifgt([-2,-1], c1, [-2,-1] - ( [-2,-1] * [-1,0] ), [-2,-1])
		& ifgt([-3,-1], c1, [-3,-1] * [-2,0], 0)
		& ifgt([-4,-1], c1, ( c2 - [-1,0] ) / c2, 1) \\
		\hline
		\hline
	\end{spreadtab}
	\caption{This Table demonstrates the progression of the $m$-DKPRG for $m = 2$ and $n = 1000$. One may ascertain that all customers eat lunch by the end of day $5$.} \label{tbl:The progression of the mDKPRG for m 2 and n 1000}
\end{table}

	The third example is more meaningful and interesting because in this case the number $n$ of agents is $10^6$, which may be thought of as representing the average case. Table \ref{tbl:The progression of the mDKPRG for m 2 and n 1000000} contains the analytical evolution of this instance. Even when confronted with a significant number of agents, the steady state utilization $1$ is rapidly achieved by the end of day $8$. Although it takes longer to reach that stage, utilization upwards of $0.98$ is established from day $2$.

\begin{table}
	\centering
	\STautoround{5}
	\begin{spreadtab}{{tabular}{cc||ccccc}}
		\hline
		\multicolumn{2}{c||}{@ Number of stops $m$} & \multicolumn{5}{c}{:={2}} \\
		\multicolumn{2}{c||}{@ Number of agents $n$} & \multicolumn{5}{c}{:={1000000}} \\
		\multicolumn{2}{c||}{@ Day ($t$)}
		& @ $n_t$ 
		& @ $VP_{t}$ 
		& @ $\overline{A_{t}^{s}} = \overline{R_{t}^{r}}$
		& @ $\overline{A_{t}^{u}} = \overline{R_{t}^{v}}$
		& @ $\overline{f_{t}}$ \\
		\hline
		@ {Start of day} & 1 & c2 & 1 & 0 & c2 & 0 \\
		@ {End of day} & 1
		& ifgt([0,-1], c1, [0,-1] * ( ( [0,-1] - c1 ) / [0,-1] )^[0,-1], 0)
		& ifgt([-1,-1], c1, ( ( [-1,-1] - c1 ) / [-1,-1] )^[-1,-1], 0)
		& ifgt([-2,-1], c1, [-2,-1] - ( [-2,-1] * [-1,0] ), [-2,-1])
		& ifgt([-3,-1], c1, [-3,-1] * [-2,0], 0)
		& ifgt([-4,-1], c1, ( c2 - [-1,0] ) / c2, 1) \\
		\hline
		@ {Start of day} & [0,-2] + 1
		& [0,-1]
		& 1
		& 0
		& [0,-1]
		& [0,-1] \\
		@ {End of day} & [0,-2] + 1
		& ifgt([0,-1], c1, [0,-1] * ( ( [0,-1] - c1 ) / [0,-1] )^[0,-1], 0)
		& ifgt([-1,-1], c1, ( ( [-1,-1] - c1 ) / [-1,-1] )^[-1,-1], 0)
		& ifgt([-2,-1], c1, [-2,-1] - ( [-2,-1] * [-1,0] ), [-2,-1])
		& ifgt([-3,-1], c1, [-3,-1] * [-2,0], 0)
		& ifgt([-4,-1], c1, ( c2 - [-1,0] ) / c2, 1) \\
		\hline
		@ {Start of day} & [0,-2] + 1
		& [0,-1]
		& 1
		& 0
		& [0,-1]
		& [0,-1] \\
		@ {End of day} & [0,-2] + 1
		& ifgt([0,-1], c1, [0,-1] * ( ( [0,-1] - c1 ) / [0,-1] )^[0,-1], 0)
		& ifgt([-1,-1], c1, ( ( [-1,-1] - c1 ) / [-1,-1] )^[-1,-1], 0)
		& ifgt([-2,-1], c1, [-2,-1] - ( [-2,-1] * [-1,0] ), [-2,-1])
		& ifgt([-3,-1], c1, [-3,-1] * [-2,0], 0)
		& ifgt([-4,-1], c1, ( c2 - [-1,0] ) / c2, 1) \\
		\hline
		@ {Start of day} & [0,-2] + 1
		& [0,-1]
		& 1
		& 0
		& [0,-1]
		& [0,-1] \\
		@ {End of day} & [0,-2] + 1
		& ifgt([0,-1], c1, [0,-1] * ( ( [0,-1] - c1 ) / [0,-1] )^[0,-1], 0)
		& ifgt([-1,-1], c1, ( ( [-1,-1] - c1 ) / [-1,-1] )^[-1,-1], 0)
		& ifgt([-2,-1], c1, [-2,-1] - ( [-2,-1] * [-1,0] ), [-2,-1])
		& ifgt([-3,-1], c1, [-3,-1] * [-2,0], 0)
		& ifgt([-4,-1], c1, ( c2 - [-1,0] ) / c2, 1) \\
		\hline
		@ {Start of day} & [0,-2] + 1
		& [0,-1]
		& 1
		& 0
		& [0,-1]
		& [0,-1] \\
		@ {End of day} & [0,-2] + 1
		& ifgt([0,-1], c1, [0,-1] * ( ( [0,-1] - c1 ) / [0,-1] )^[0,-1], 0)
		& ifgt([-1,-1], c1, ( ( [-1,-1] - c1 ) / [-1,-1] )^[-1,-1], 0)
		& ifgt([-2,-1], c1, [-2,-1] - ( [-2,-1] * [-1,0] ), [-2,-1])
		& ifgt([-3,-1], c1, [-3,-1] * [-2,0], 0)
		& ifgt([-4,-1], c1, ( c2 - [-1,0] ) / c2, 1) \\
		\hline
		@ {Start of day} & [0,-2] + 1
		& [0,-1]
		& 1
		& 0
		& [0,-1]
		& [0,-1] \\
		@ {End of day} & [0,-2] + 1
		& ifgt([0,-1], c1, [0,-1] * ( ( [0,-1] - c1 ) / [0,-1] )^[0,-1], 0)
		& ifgt([-1,-1], c1, ( ( [-1,-1] - c1 ) / [-1,-1] )^[-1,-1], 0)
		& ifgt([-2,-1], c1, [-2,-1] - ( [-2,-1] * [-1,0] ), [-2,-1])
		& ifgt([-3,-1], c1, [-3,-1] * [-2,0], 0)
		& ifgt([-4,-1], c1, ( c2 - [-1,0] ) / c2, 1) \\
		\hline
		@ {Start of day} & [0,-2] + 1
		& [0,-1]
		& 1
		& 0
		& [0,-1]
		& [0,-1] \\
		@ {End of day} & [0,-2] + 1
		& ifgt([0,-1], c1, [0,-1] * ( ( [0,-1] - c1 ) / [0,-1] )^[0,-1], 0)
		& ifgt([-1,-1], c1, ( ( [-1,-1] - c1 ) / [-1,-1] )^[-1,-1], 0)
		& ifgt([-2,-1], c1, [-2,-1] - ( [-2,-1] * [-1,0] ), [-2,-1])
		& ifgt([-3,-1], c1, [-3,-1] * [-2,0], 0)
		& ifgt([-4,-1], c1, ( c2 - [-1,0] ) / c2, 1) \\
		\hline
		@ {Start of day} & [0,-2] + 1 
		& [0,-1]
		& 1
		& 0
		& [0,-1]
		& [0,-1] \\
		@ {End of day} & [0,-2] + 1
		& ifgt([0,-1], c1, [0,-1] * ( ( [0,-1] - c1 ) / [0,-1] )^[0,-1], 0)
		& ifgt([-1,-1], c1, ( ( [-1,-1] - c1 ) / [-1,-1] )^[-1,-1], 0)
		& ifgt([-2,-1], c1, [-2,-1] - ( [-2,-1] * [-1,0] ), [-2,-1])
		& ifgt([-3,-1], c1, [-3,-1] * [-2,0], 0)
		& ifgt([-4,-1], c1, ( c2 - [-1,0] ) / c2, 1) \\
		\hline
		\hline
	\end{spreadtab}
	\caption{This Table demonstrates the progression of the $m$-DKPRG for $m = 2$ and $n = 1000000$. Although in this case it takes longer to reach the steady state, utilization upwards of $0.98$ is established from day $2$.} \label{tbl:The progression of the mDKPRG for m 2 and n 1000000}
\end{table}

	The fourth example is instructive about the behavior of our strategy when a large number of agents is involved. In this case the number $n$ of agents is $10^9$ and, unsurprisingly, it takes $11$ days to reach the steady state utilization $1$. All the details of the progression of this game are given in Table \ref{tbl:The progression of the mDKPRG for m 2 and n 1000000000}. Careful observation of the data confirms a major characteristic of our distributed game: for $m = 2$ steps the first day utilization is at least $0.86$ and it goes over $0.98$ from day $2$.

\begin{table}
	\centering
	\STautoround{5}
	\begin{spreadtab}{{tabular}{cc||ccccc}}
		\hline
		\multicolumn{2}{c||}{@ Number of stops $m$} & \multicolumn{5}{c}{:={2}} \\
		\multicolumn{2}{c||}{@ Number of agents $n$} & \multicolumn{5}{c}{:={1000000000}} \\
		\multicolumn{2}{c||}{@ Day ($t$)}
		& @ $n_t$ 
		& @ $VP_{t}$ 
		& @ $\overline{A_{t}^{s}} = \overline{R_{t}^{r}}$
		& @ $\overline{A_{t}^{u}} = \overline{R_{t}^{v}}$
		& @ $\overline{f_{t}}$ \\
		\hline
		@ {Start of day} & 1 & c2 & 1 & 0 & c2 & 0 \\
		@ {End of day} & 1
		& ifgt([0,-1], c1, [0,-1] * ( ( [0,-1] - c1 ) / [0,-1] )^[0,-1], 0)
		& ifgt([-1,-1], c1, ( ( [-1,-1] - c1 ) / [-1,-1] )^[-1,-1], 0)
		& ifgt([-2,-1], c1, [-2,-1] - ( [-2,-1] * [-1,0] ), [-2,-1])
		& ifgt([-3,-1], c1, [-3,-1] * [-2,0], 0)
		& ifgt([-4,-1], c1, ( c2 - [-1,0] ) / c2, 1) \\
		\hline
		@ {Start of day} & [0,-2] + 1
		& [0,-1]
		& 1
		& 0
		& [0,-1]
		& [0,-1] \\
		@ {End of day} & [0,-2] + 1
		& ifgt([0,-1], c1, [0,-1] * ( ( [0,-1] - c1 ) / [0,-1] )^[0,-1], 0)
		& ifgt([-1,-1], c1, ( ( [-1,-1] - c1 ) / [-1,-1] )^[-1,-1], 0)
		& ifgt([-2,-1], c1, [-2,-1] - ( [-2,-1] * [-1,0] ), [-2,-1])
		& ifgt([-3,-1], c1, [-3,-1] * [-2,0], 0)
		& ifgt([-4,-1], c1, ( c2 - [-1,0] ) / c2, 1) \\
		\hline
		@ {Start of day} & [0,-2] + 1
		& [0,-1]
		& 1
		& 0
		& [0,-1]
		& [0,-1] \\
		@ {End of day} & [0,-2] + 1
		& ifgt([0,-1], c1, [0,-1] * ( ( [0,-1] - c1 ) / [0,-1] )^[0,-1], 0)
		& ifgt([-1,-1], c1, ( ( [-1,-1] - c1 ) / [-1,-1] )^[-1,-1], 0)
		& ifgt([-2,-1], c1, [-2,-1] - ( [-2,-1] * [-1,0] ), [-2,-1])
		& ifgt([-3,-1], c1, [-3,-1] * [-2,0], 0)
		& ifgt([-4,-1], c1, ( c2 - [-1,0] ) / c2, 1) \\
		\hline
		@ {Start of day} & [0,-2] + 1
		& [0,-1]
		& 1
		& 0
		& [0,-1]
		& [0,-1] \\
		@ {End of day} & [0,-2] + 1
		& ifgt([0,-1], c1, [0,-1] * ( ( [0,-1] - c1 ) / [0,-1] )^[0,-1], 0)
		& ifgt([-1,-1], c1, ( ( [-1,-1] - c1 ) / [-1,-1] )^[-1,-1], 0)
		& ifgt([-2,-1], c1, [-2,-1] - ( [-2,-1] * [-1,0] ), [-2,-1])
		& ifgt([-3,-1], c1, [-3,-1] * [-2,0], 0)
		& ifgt([-4,-1], c1, ( c2 - [-1,0] ) / c2, 1) \\
		\hline
		@ {Start of day} & [0,-2] + 1
		& [0,-1]
		& 1
		& 0
		& [0,-1]
		& [0,-1] \\
		@ {End of day} & [0,-2] + 1
		& ifgt([0,-1], c1, [0,-1] * ( ( [0,-1] - c1 ) / [0,-1] )^[0,-1], 0)
		& ifgt([-1,-1], c1, ( ( [-1,-1] - c1 ) / [-1,-1] )^[-1,-1], 0)
		& ifgt([-2,-1], c1, [-2,-1] - ( [-2,-1] * [-1,0] ), [-2,-1])
		& ifgt([-3,-1], c1, [-3,-1] * [-2,0], 0)
		& ifgt([-4,-1], c1, ( c2 - [-1,0] ) / c2, 1) \\
		\hline
		@ {Start of day} & [0,-2] + 1
		& [0,-1]
		& 1
		& 0
		& [0,-1]
		& [0,-1] \\
		@ {End of day} & [0,-2] + 1
		& ifgt([0,-1], c1, [0,-1] * ( ( [0,-1] - c1 ) / [0,-1] )^[0,-1], 0)
		& ifgt([-1,-1], c1, ( ( [-1,-1] - c1 ) / [-1,-1] )^[-1,-1], 0)
		& ifgt([-2,-1], c1, [-2,-1] - ( [-2,-1] * [-1,0] ), [-2,-1])
		& ifgt([-3,-1], c1, [-3,-1] * [-2,0], 0)
		& ifgt([-4,-1], c1, ( c2 - [-1,0] ) / c2, 1) \\
		\hline
		@ {Start of day} & [0,-2] + 1
		& [0,-1]
		& 1
		& 0
		& [0,-1]
		& [0,-1] \\
		@ {End of day} & [0,-2] + 1
		& ifgt([0,-1], c1, [0,-1] * ( ( [0,-1] - c1 ) / [0,-1] )^[0,-1], 0)
		& ifgt([-1,-1], c1, ( ( [-1,-1] - c1 ) / [-1,-1] )^[-1,-1], 0)
		& ifgt([-2,-1], c1, [-2,-1] - ( [-2,-1] * [-1,0] ), [-2,-1])
		& ifgt([-3,-1], c1, [-3,-1] * [-2,0], 0)
		& ifgt([-4,-1], c1, ( c2 - [-1,0] ) / c2, 1) \\
		\hline
		@ {Start of day} & [0,-2] + 1
		& [0,-1]
		& 1
		& 0
		& [0,-1]
		& [0,-1] \\
		@ {End of day} & [0,-2] + 1
		& ifgt([0,-1], c1, [0,-1] * ( ( [0,-1] - c1 ) / [0,-1] )^[0,-1], 0)
		& ifgt([-1,-1], c1, ( ( [-1,-1] - c1 ) / [-1,-1] )^[-1,-1], 0)
		& ifgt([-2,-1], c1, [-2,-1] - ( [-2,-1] * [-1,0] ), [-2,-1])
		& ifgt([-3,-1], c1, [-3,-1] * [-2,0], 0)
		& ifgt([-4,-1], c1, ( c2 - [-1,0] ) / c2, 1) \\
		\hline
		@ {Start of day} & [0,-2] + 1
		& [0,-1]
		& 1
		& 0
		& [0,-1]
		& [0,-1] \\
		@ {End of day} & [0,-2] + 1
		& ifgt([0,-1], c1, [0,-1] * ( ( [0,-1] - c1 ) / [0,-1] )^[0,-1], 0)
		& ifgt([-1,-1], c1, ( ( [-1,-1] - c1 ) / [-1,-1] )^[-1,-1], 0)
		& ifgt([-2,-1], c1, [-2,-1] - ( [-2,-1] * [-1,0] ), [-2,-1])
		& ifgt([-3,-1], c1, [-3,-1] * [-2,0], 0)
		& ifgt([-4,-1], c1, ( c2 - [-1,0] ) / c2, 1) \\
		\hline
		@ {Start of day} & [0,-2] + 1
		& [0,-1]
		& 1
		& 0
		& [0,-1]
		& [0,-1] \\
		@ {End of day} & [0,-2] + 1
		& ifgt([0,-1], c1, [0,-1] * ( ( [0,-1] - c1 ) / [0,-1] )^[0,-1], 0)
		& ifgt([-1,-1], c1, ( ( [-1,-1] - c1 ) / [-1,-1] )^[-1,-1], 0)
		& ifgt([-2,-1], c1, [-2,-1] - ( [-2,-1] * [-1,0] ), [-2,-1])
		& ifgt([-3,-1], c1, [-3,-1] * [-2,0], 0)
		& ifgt([-4,-1], c1, ( c2 - [-1,0] ) / c2, 1) \\
		\hline
		@ {Start of day} & [0,-2] + 1 
		& [0,-1]
		& 1
		& 0
		& [0,-1]
		& [0,-1] \\
		@ {End of day} & [0,-2] + 1
		& ifgt([0,-1], c1, [0,-1] * ( ( [0,-1] - c1 ) / [0,-1] )^[0,-1], 0)
		& ifgt([-1,-1], c1, ( ( [-1,-1] - c1 ) / [-1,-1] )^[-1,-1], 0)
		& ifgt([-2,-1], c1, [-2,-1] - ( [-2,-1] * [-1,0] ), [-2,-1])
		& ifgt([-3,-1], c1, [-3,-1] * [-2,0], 0)
		& ifgt([-4,-1], c1, ( c2 - [-1,0] ) / c2, 1) \\
		\hline
		\hline
	\end{spreadtab}
	\caption{This Table demonstrates the progression of the $m$-DKPRG for $m = 2$ and $n = 1000000000$. One can see that the first day utilization is at least $0.86$ and it goes over $0.98$ from day $2$.} \label{tbl:The progression of the mDKPRG for m 2 and n 1000000000}
\end{table}

	It is quite straightforward to convince ourselves that playing a $3$-stop game is better than playing a a $2$-stop game. A precise quantitative analysis of the resulting advantages can be performed by considering the exact formulas (\ref{eq:Vacancy Probability VPt}) - (\ref{eq:Expected Utilization on Day t}). Nonetheless, we believe it is expedient to showcase the difference with the following example. The present example resembles the previous one in that the number of agents is the same, namely $10^9$. However, this time each agent may visit up to \emph{three} restaurants if need be. Such an instance, with a large number of agents, can serve as the best demonstration of the dramatic improvement that can be obtained by an increase in the number of steps. Indeed, the data in the Table \ref{tbl:The progression of the mDKPRG for m 3 and n 1000000000} corroborate this expectation, as one can now see that all restaurants are utilized by the end of day $8$, compared to day $11$ before, the utilization at the end of first day is already up to an impressive $0.95$ and becomes $1$, for all practical purposes, at the end of day $6$. This last example can be considered as a compelling argument that advocates the importance of topological analysis for the network of restaurants.

\begin{table}
	\centering
	\STautoround{5}
	\begin{spreadtab}{{tabular}{cc||ccccc}}
		\hline
		\multicolumn{2}{c||}{@ Number of stops $m$} & \multicolumn{5}{c}{:={3}} \\
		\multicolumn{2}{c||}{@ Number of agents $n$} & \multicolumn{5}{c}{:={1000000000}} \\
		\multicolumn{2}{c||}{@ Day ($t$)}
		& @ $n_t$ 
		& @ $VP_{t}$ 
		& @ $\overline{A_{t}^{s}} = \overline{R_{t}^{r}}$
		& @ $\overline{A_{t}^{u}} = \overline{R_{t}^{v}}$
		& @ $\overline{f_{t}}$ \\
		\hline
		@ {Start of day} & 1 & c2 & 1 & 0 & c2 & 0 \\
		@ {End of day} & 1
		& ifgt([0,-1], c1, [0,-1] * ( ( [0,-1] - c1 ) / [0,-1] )^[0,-1], 0)
		& ifgt([-1,-1], c1, ( ( [-1,-1] - c1 ) / [-1,-1] )^[-1,-1], 0)
		& ifgt([-2,-1], c1, [-2,-1] - ( [-2,-1] * [-1,0] ), [-2,-1])
		& ifgt([-3,-1], c1, [-3,-1] * [-2,0], 0)
		& ifgt([-4,-1], c1, ( c2 - [-1,0] ) / c2, 1) \\
		\hline
		@ {Start of day} & [0,-2] + 1
		& [0,-1]
		& 1
		& 0
		& [0,-1]
		& [0,-1] \\
		@ {End of day} & [0,-2] + 1
		& ifgt([0,-1], c1, [0,-1] * ( ( [0,-1] - c1 ) / [0,-1] )^[0,-1], 0)
		& ifgt([-1,-1], c1, ( ( [-1,-1] - c1 ) / [-1,-1] )^[-1,-1], 0)
		& ifgt([-2,-1], c1, [-2,-1] - ( [-2,-1] * [-1,0] ), [-2,-1])
		& ifgt([-3,-1], c1, [-3,-1] * [-2,0], 0)
		& ifgt([-4,-1], c1, ( c2 - [-1,0] ) / c2, 1) \\
		\hline
		@ {Start of day} & [0,-2] + 1
		& [0,-1]
		& 1
		& 0
		& [0,-1]
		& [0,-1] \\
		@ {End of day} & [0,-2] + 1
		& ifgt([0,-1], c1, [0,-1] * ( ( [0,-1] - c1 ) / [0,-1] )^[0,-1], 0)
		& ifgt([-1,-1], c1, ( ( [-1,-1] - c1 ) / [-1,-1] )^[-1,-1], 0)
		& ifgt([-2,-1], c1, [-2,-1] - ( [-2,-1] * [-1,0] ), [-2,-1])
		& ifgt([-3,-1], c1, [-3,-1] * [-2,0], 0)
		& ifgt([-4,-1], c1, ( c2 - [-1,0] ) / c2, 1) \\
		\hline
		@ {Start of day} & [0,-2] + 1
		& [0,-1]
		& 1
		& 0
		& [0,-1]
		& [0,-1] \\
		@ {End of day} & [0,-2] + 1
		& ifgt([0,-1], c1, [0,-1] * ( ( [0,-1] - c1 ) / [0,-1] )^[0,-1], 0)
		& ifgt([-1,-1], c1, ( ( [-1,-1] - c1 ) / [-1,-1] )^[-1,-1], 0)
		& ifgt([-2,-1], c1, [-2,-1] - ( [-2,-1] * [-1,0] ), [-2,-1])
		& ifgt([-3,-1], c1, [-3,-1] * [-2,0], 0)
		& ifgt([-4,-1], c1, ( c2 - [-1,0] ) / c2, 1) \\
		\hline
		@ {Start of day} & [0,-2] + 1
		& [0,-1]
		& 1
		& 0
		& [0,-1]
		& [0,-1] \\
		@ {End of day} & [0,-2] + 1
		& ifgt([0,-1], c1, [0,-1] * ( ( [0,-1] - c1 ) / [0,-1] )^[0,-1], 0)
		& ifgt([-1,-1], c1, ( ( [-1,-1] - c1 ) / [-1,-1] )^[-1,-1], 0)
		& ifgt([-2,-1], c1, [-2,-1] - ( [-2,-1] * [-1,0] ), [-2,-1])
		& ifgt([-3,-1], c1, [-3,-1] * [-2,0], 0)
		& ifgt([-4,-1], c1, ( c2 - [-1,0] ) / c2, 1) \\
		\hline
		@ {Start of day} & [0,-2] + 1
		& [0,-1]
		& 1
		& 0
		& [0,-1]
		& [0,-1] \\
		@ {End of day} & [0,-2] + 1
		& ifgt([0,-1], c1, [0,-1] * ( ( [0,-1] - c1 ) / [0,-1] )^[0,-1], 0)
		& ifgt([-1,-1], c1, ( ( [-1,-1] - c1 ) / [-1,-1] )^[-1,-1], 0)
		& ifgt([-2,-1], c1, [-2,-1] - ( [-2,-1] * [-1,0] ), [-2,-1])
		& ifgt([-3,-1], c1, [-3,-1] * [-2,0], 0)
		& ifgt([-4,-1], c1, ( c2 - [-1,0] ) / c2, 1) \\
		\hline
		@ {Start of day} & [0,-2] + 1
		& [0,-1]
		& 1
		& 0
		& [0,-1]
		& [0,-1] \\
		@ {End of day} & [0,-2] + 1
		& ifgt([0,-1], c1, [0,-1] * ( ( [0,-1] - c1 ) / [0,-1] )^[0,-1], 0)
		& ifgt([-1,-1], c1, ( ( [-1,-1] - c1 ) / [-1,-1] )^[-1,-1], 0)
		& ifgt([-2,-1], c1, [-2,-1] - ( [-2,-1] * [-1,0] ), [-2,-1])
		& ifgt([-3,-1], c1, [-3,-1] * [-2,0], 0)
		& ifgt([-4,-1], c1, ( c2 - [-1,0] ) / c2, 1) \\
		\hline
		@ {Start of day} & [0,-2] + 1 
		& [0,-1]
		& 1
		& 0
		& [0,-1]
		& [0,-1] \\
		@ {End of day} & [0,-2] + 1
		& ifgt([0,-1], c1, [0,-1] * ( ( [0,-1] - c1 ) / [0,-1] )^[0,-1], 0)
		& ifgt([-1,-1], c1, ( ( [-1,-1] - c1 ) / [-1,-1] )^[-1,-1], 0)
		& ifgt([-2,-1], c1, [-2,-1] - ( [-2,-1] * [-1,0] ), [-2,-1])
		& ifgt([-3,-1], c1, [-3,-1] * [-2,0], 0)
		& ifgt([-4,-1], c1, ( c2 - [-1,0] ) / c2, 1) \\
		\hline
		\hline
	\end{spreadtab}
	\caption{This Table demonstrates the progression of the $m$-DKPRG for $m = 3$ and $n = 1000000000$.} \label{tbl:The progression of the mDKPRG for m 3 and n 1000000000}
\end{table}

	The above examples were studied using the exact formulas (\ref{eq:Vacancy Probability VPt}) - (\ref{eq:Expected Utilization on Day t}). Tables \ref{tbl:The progression of the mDKPRG for m 2 and n 100} - \ref{tbl:The progression of the mDKPRG for m 3 and n 1000000000} reflect the daily evolution of the above five instances of the $m$-DKPRG according to rigorous mathematical description provided by formulas (\ref{eq:Vacancy Probability VPt}) - (\ref{eq:Expected Utilization on Day t}). The next Figure \ref{fig:Exact Expected Utilization for m=2,3} is a graphical representation of the exact utilization $\overline{f_{t}}$ from all the previous examples, as shown in the Tables \ref{tbl:The progression of the mDKPRG for m 2 and n 100} - \ref{tbl:The progression of the mDKPRG for m 3 and n 1000000000}. In this Figure, the generally excellent behavior of this scheme can be easily verified. We point out the rapid convergence to the steady state in a matter of few days and, especially, the superiority of the three stop policy. The latter achieves $0.95$ utilization from the first day and above $0.99$ from the second day.

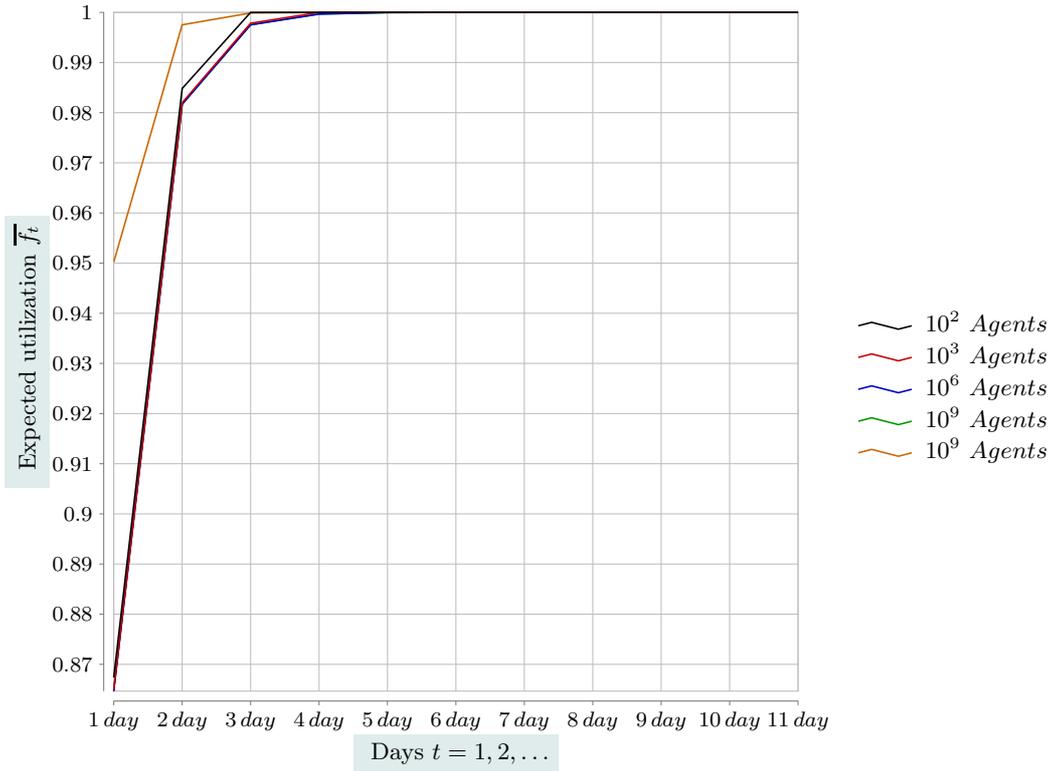
\begin{figure}
	\centering
	\begin{tikzpicture} [scale = 0.75]
		\datavisualization [
		scientific axes = clean,
		x axis = { attribute = Days, length = 12cm, ticks = {tick unit = day, step = 1}, label = { [node style = { fill = WordVeryLightTeal } ] { Days $t = 1, 2, \ldots$ } } },
		y axis = { attribute = Utilization, length = 12cm, ticks = {step = 0.01}, label = { [node style = { fill = WordVeryLightTeal } ] { Expected utilization $\overline{f_{t}}$ } } },
		all axes = grid,
		style sheet = strong colors,
		visualize as line/.list = {a, b, c, d, e},
		a = { label in legend = {text = $10^2 \ Agents$} },
		b = { label in legend = {text = $10^3 \ Agents$} },
		c = { label in legend = {text = $10^6 \ Agents$} },
		d = { label in legend = {text = $10^9 \ Agents$} },
		e = { label in legend = {text = $10^9 \ Agents$} }
		]
		data [set = a] {
			Days, Utilization
			1, 0.8673804
			2, 0.9848263
			3, 1.0
			4, 1.0
			5, 1.0
			6, 1.0
			7, 1.0
			8, 1.0
			9, 1.0
			10, 1.0
			11, 1.0
		}
		data [set = b] {
			Days, Utilization
			1, 0.8649355
			2, 0.9819923
			3, 0.9978386
			4, 0.9999921
			5, 1.0
			6, 1.0
			7, 1.0
			8, 1.0
			9, 1.0
			10, 1.0
			11, 1.0
		}
		data [set = c] {
			Days, Utilization
			1, 0.864665
			2, 0.9816847
			3, 0.9975216
			4, 0.9996649
			5, 0.9999549
			6, 0.9999942
			7, 0.9999995
			8, 1.0
			9, 1.0
			10, 1.0
			11, 1.0
		}
		data [set = d] {
			Days, Utilization
			1, 0.8646647
			2, 0.9816844
			3, 0.9975212
			4, 0.9996645
			5, 0.9999546
			6, 0.9999939
			7, 0.9999992
			8, 0.9999999
			9, 1.0
			10, 1.0
			11, 1.0
		}
		data [set = e] {
			Days, Utilization
			1, 0.9502129
			2, 0.9975212
			3, 0.9998766
			4, 0.9999939
			5, 0.9999997
			6, 1.0
			7, 1.0
			8, 1.0
			9, 1.0
			10, 1.0
			11, 1.0
		};
	\end{tikzpicture}
	\caption{This figure depicts the exact expected utilization $\overline{f_{t}}$, as given by equation (\ref{eq:Expected Utilization on Day t}), for all five instances of the $m$-DKPRG studied in Tables \ref{tbl:The progression of the mDKPRG for m 2 and n 100} - \ref{tbl:The progression of the mDKPRG for m 3 and n 1000000000}.} \label{fig:Exact Expected Utilization for m=2,3}
\end{figure}

	The above remarks must not diminish the value of the approximate formulas (\ref{eq:Approximation of Expected Number of Active Agents on Day t+1}) - (\ref{eq:Approximate Expected Utilization on Day t}). Their value lies on the fact that they can provide easy to compute and particularly good approximations for large $n$. A simple comparison of Figure \ref{fig:Exact Expected Utilization for m=2,3} to the approximations shown in Figure \ref{fig:Approximate Expected Utilization for m=2}, which corresponds to the case $m = 2$, and in Figure \ref{fig:Approximate Expected Utilization for m=3}, which depicts the case where $m = 3$, ascertains their accuracy.

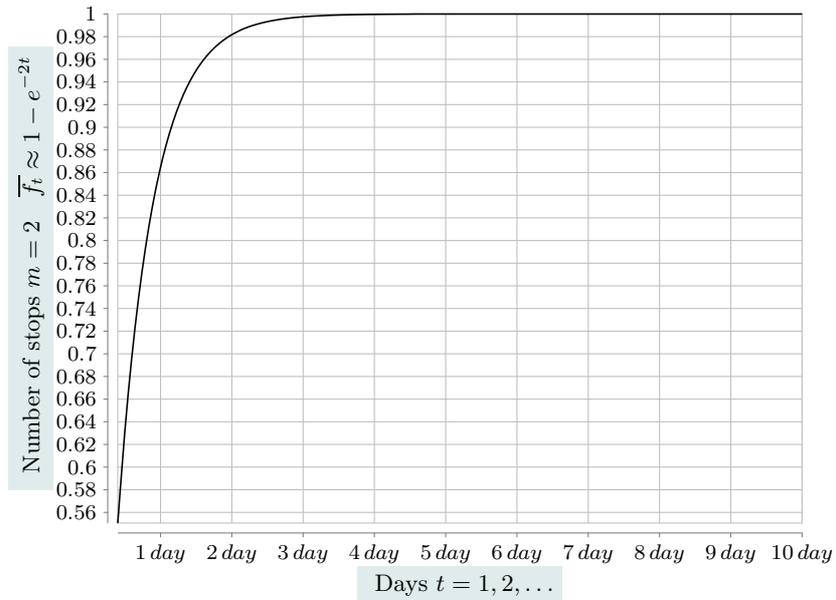
\begin{figure}
	\centering
	\begin{tikzpicture} [scale = 0.75]
		\datavisualization [
		scientific axes = clean,
		x axis = { label = { [node style = { fill = WordVeryLightTeal } ] { Days $t = 1, 2, \ldots$ } }, length = 12cm, ticks = {tick unit = day, step = 1} },
		y axis = { label = { [node style = { fill = WordVeryLightTeal } ] { Number of stops $m = 2 \quad \overline{f_{t}} \approx 1 - e^{ - 2 t }$ } }, length = 9cm, ticks = {step = 0.02} },
		all axes = {grid},
		visualize as smooth line,
		data/format = function]
		data
		{
			var x : interval [0.4 : 10] samples 100;
			func y = 1 - exp( - 2 * \value x ) ;
		};
	\end{tikzpicture}
	\caption{This figure depicts the approximate expected utilization $\overline{f_{t}} \approx 1 - e^{ - 2 t }$ as given by equation (\ref{eq:Approximate Expected Utilization on Day t}) for $m = 2$.} \label{fig:Approximate Expected Utilization for m=2}
\end{figure}

\begin{figure}
	\centering
	\begin{tikzpicture} [scale = 0.75]
		\datavisualization [
		scientific axes = clean,
		x axis = { label = { [node style = { fill = WordVeryLightTeal } ] { Days $t = 1, 2, \ldots$ } }, length = 12cm, ticks = {tick unit = day, step = 1} },
		y axis = { label = { [node style = { fill = WordVeryLightTeal } ] { Number of stops $m = 3 \quad \overline{f_{t}} \approx 1 - e^{ - 3 t }$ } }, length = 9cm, ticks = {step = 0.02} },
		all axes = {grid},
		visualize as smooth line,
		data/format = function]
		data
		{
			var x : interval [0.4 : 10] samples 100;
			func y = 1 - exp( - 3 * \value x ) ;
		};
	\end{tikzpicture}
	\caption{This figure depicts the approximate expected utilization $\overline{f_{t}} \approx 1 - e^{ - 3 t }$ as given by equation (\ref{eq:Approximate Expected Utilization on Day t}) for $m = 3$.} \label{fig:Approximate Expected Utilization for m=3}
\end{figure}
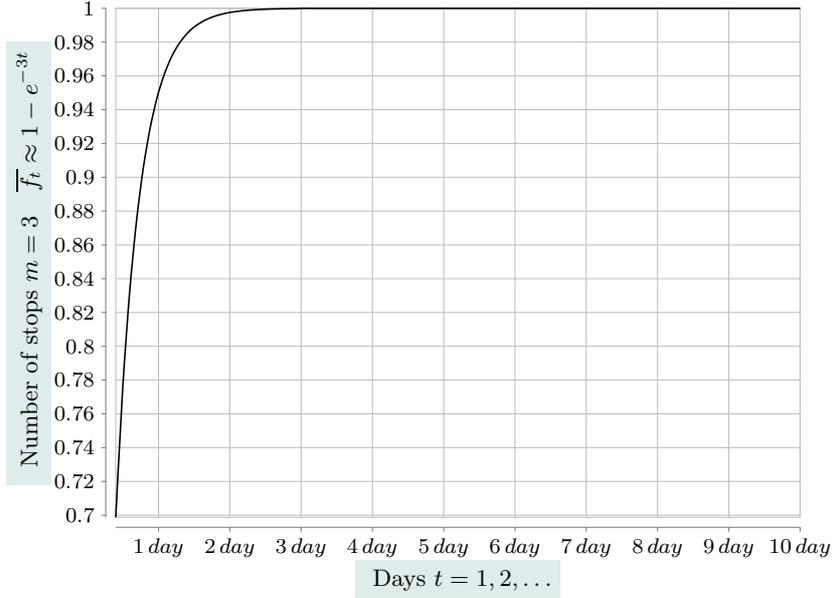

\section{Conclusion} \label{sec:Conclusion}

	This work explored a completely new angle of the Kolkata Paise Restaurant Problem. The topological layout of the restaurants takes center stage in this new paradigm. Initially, we explicitly stated certain assumptions that are implicitly present in the standard formulation of the game. Having done that, we undertook the radical step to go past them and create an entirely new setting. The critical examination of the topological setting of the game unavoidably enhanced our perception regarding the locations of the restaurants and suggested a more realistic topological layout. We argued that their uniform distribution in the entire city area is the most logical, fair, and probable situation. As a result, we defined a new version of the game that is spatially distributed and, for this, is is aptly named the Distributed Kolkate Paise Restaurant Game (DKPRG).

	The uniform probabilistic distribution of the restaurants enabled us to rigorously prove that, as their number $n$ increases, the restaurants \emph{get closer} and the distance between adjacent restaurants decreases. In such a network, every customer has the opportunity to pass through more than one restaurants within the allowed time window. The agents now become travelling salesmen and this led us to suggest the innovative idea that TSP can be used to increase the chances of success in this game. We propose that each agent should use metaheuristics to solve her personalized TSP because metaheuristics produce near-optimal solutions very fast and as such can be easily used in practice. This culminated in the development of a new and more efficient strategy that achieves greater utilization.

	After rigorously formulating DKPRG, we proved completely general formulas that assert the increase in utilization of our scheme. We established that utilization ranging from $0.85$ to $0.95$ is achievable. This was shown in great detail in Tables \ref{tbl:The progression of the mDKPRG for m 2 and n 100} - \ref{tbl:The progression of the mDKPRG for m 3 and n 1000000000}, which depict the daily progress of characteristic instances of the DKPRG according to the rigorous mathematical description provided by the exact formulas (\ref{eq:Vacancy Probability VPt}) - (\ref{eq:Expected Utilization on Day t}). Apart from the exact formulas, we also derived the approximate formulas (\ref{eq:Approximation of Expected Number of Active Agents on Day t+1}) - (\ref{eq:Approximate Expected Utilization on Day t}). They can be quite useful because they are considerably easier to compute and are exceedingly good approximations for large $n$. This fact is easily corroborated by comparing Figure \ref{fig:Exact Expected Utilization for m=2,3} to the approximations shown in Figures \ref{fig:Approximate Expected Utilization for m=2} and \ref{fig:Approximate Expected Utilization for m=3}. Let us remark that the derived equations generalize previously presented formulas in the literature.

	It is worth mentioning that the fact that our strategy exhibits very rapid convergence to the steady state of utilization $1.0$ can be potentially used to address the following situation. An issue that remains and is common to almost all works in the literature is the simple matter that a near optimal utilization may not, in general, be optimal for every agent individually. A socially efficient outcome where every agent eats lunch and every restaurant gets a customer to serve, is not necessarily optimal for the individual customer, in the sense that an agent may get served in a restaurant of low preference. A possible solution to this might be to reset the game periodically. We expect that adopting a reset period, i.e., setting a specific period of days, after which the system is reset and the game starts from scratch, may alleviate this drawback. In any event, this idea for a future work will require further study and experimental evaluation of its usefulness. Finally, another possible direction for future work could include extensive experimental tests and further investigation of other versions of TSP. For instance, that there exists a more restrictive version of the TSP, the Travelling Salesman Problem with Time Windows (TSP-TW). TSP-TW is a constrained version of TSP in which the salesman must visit the cities within a specific time window. This version is even more complicated and difficult to solve. However, the inherent time constraints built-in the TSP-TW may provide for an even more realistic modeling of the DKPRG, so it is a research avenue that we believe is worth pursuing.

\bibliographystyle{ieeetr}
\bibliography{DKPRGBibliography}

\end{document}